\documentclass{amsart}[12pt]
\usepackage{latexsym}
\usepackage{amsmath}
\usepackage{amsfonts}
\usepackage{amssymb}
\usepackage{graphicx}
\usepackage{amsthm}
\usepackage{hyperref}
\usepackage[titletoc,toc]{appendix}

\theoremstyle{plain}
\newtheorem{theorem}{Theorem}[section]
\newtheorem{proposition}[theorem]{Proposition}
\newtheorem{lemma}[theorem]{Lemma}

\theoremstyle{definition}
\newtheorem{definition}[theorem]{Definition}

\newtheorem{remark}[theorem]{Remark}

\numberwithin{equation}{section}

\allowdisplaybreaks[1]

\swapnumbers
\begin{document}

\title[]{On the backward stability of the \\
Schwarzschild
black hole singularity}

\author{Grigorios Fournodavlos}

\thanks{Mathematics Department, University of Toronto. Email: grifour@math.toronto.edu}

\date{}

\begin{abstract}
We study the backwards-in-time stability of the Schwarzschild singularity from a dynamical PDE point of view. More precisely, considering a spacelike hypersurface $\Sigma_0$ in the interior of the black hole region, tangent to the singular hypersurface $\{r=0\}$ at a single sphere, we study the problem of perturbing the Schwarzschild data on $\Sigma_0$ and solving the Einstein vacuum equations backwards in time.
We obtain a local well-posedness result for small perturbations lying in certain weighted Sobolev spaces. No symmetry assumptions are imposed. The perturbed spacetimes all have a singularity at a ``collapsed'' sphere on $\Sigma_0$, where the leading asymptotics of the curvature and the metric match those of their Schwarzschild counterparts to a suitably high order.
As in the Schwarzschild backward evolution, the pinched initial hypersurface $\Sigma_0$ `opens up' instantly, becoming a smooth spacelike (cylindrical) hypersurface.
This result thus yields classes of examples of non-symmetric vacuum spacetimes, evolving forward-in-time from smooth initial data, which form a Schwarzschild type singularity at a collapsed sphere. We rely on a precise asymptotic analysis of the Schwarzschild geometry near the singularity which turns out to be at the threshold that our energy methods can handle.
\end{abstract}

\maketitle

\tableofcontents

\parskip = 8 pt

\section{Introduction}

It is well-known (cf. Birkhoff's theorem \cite{DafRod}) that the only spherically symmetric solution $(\mathcal{M}^{1+3},g)$ to the Einstein vacuum equations (EVE)
\begin{align}\label{EVE}
\text{Ric}_{ab}(g)=0,
\end{align}
is the celebrated Schwarzschild spacetime. It was in fact the first non-trivial solution to the Einstein field equations to be discovered \cite{DafRod}. In Kruskal (null) $u,v$ coordinates the maximally extended metric reads
\begin{align}\label{Schw}
{^Sg}=-\Omega^2dudv+r^2(d\theta^2+\sin^2\theta d\phi^2),
\end{align}
where $\Omega^2=\frac{32M^3}{r}e^{-\frac{r}{2M}}$, $M>0$, and the radius function $r$ is given implicitly by
\begin{align}\label{uv}
uv=(1-\frac{r}{2M})e^{\frac{r}{2M}}.
\end{align}
\begin{figure}[h!]
  \centering
    \includegraphics[width=0.6\textwidth]{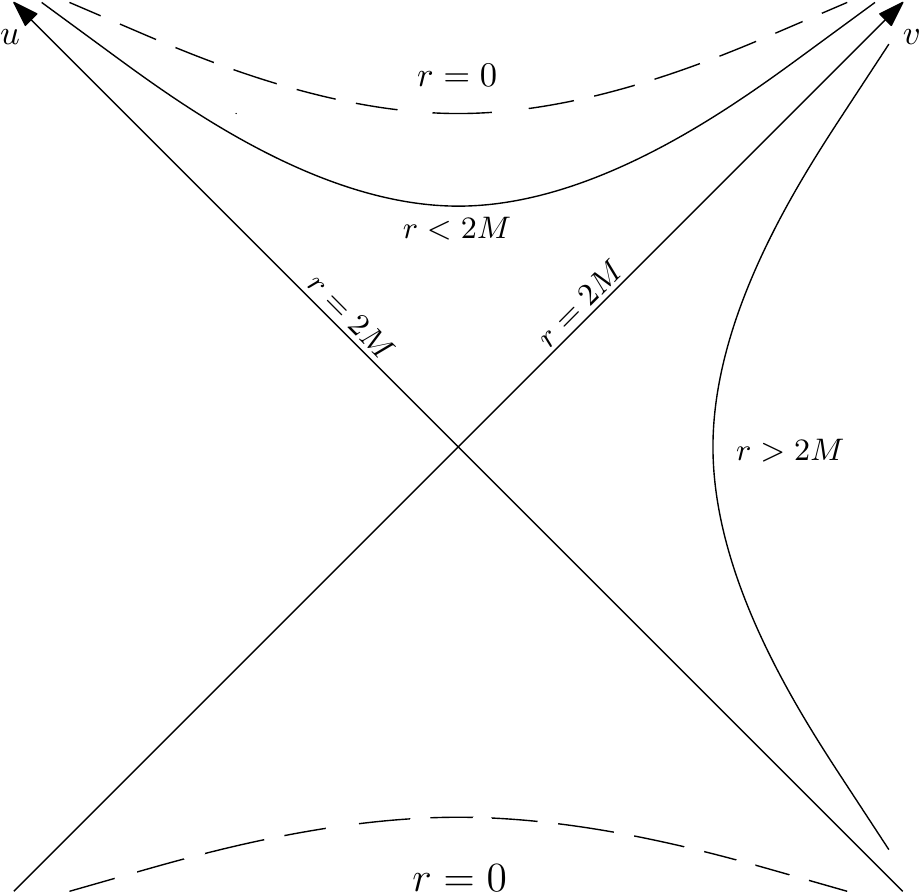}
      \caption{The Kruskal plane.}
\label{Krusk}
\end{figure}
Here the underlying manifold $^S\mathcal{M}^{1+3}$ is endowed with the differential structure of $\mathcal{U}\times\mathbb{S}^2$, where $\mathcal{U}$ is the open subset $\{uv<1\}$ in the $uv$ plane; see Figure \ref{Krusk}.
The spacetime has an essential curvature singularity at $r=0$, (the future component of) which is contained in the interior of the black hole region, the quadrant $u>0,v>0$.
In fact, a short computation shows that the Gauss curvature of the $uv$-plane equals
\begin{align}\label{K}
{^S\text{K}}=\frac{2M}{r^3}
\end{align}
and hence the manifold is $C^2$ inextendible past $r=0$. An interesting feature of this singularity is its spacelike character, that is, it can be viewed as a spacelike hypersurface.

Yet another interesting feature of the Schwarzschild singularity is its unstable nature from the evolutionary dynamical point of view. To illustrate this consider 
a global spacelike Cauchy hypersurface $\Sigma^3$ in Schwarzschild (Figure \ref{SigmaIVP}). An initial data set for the EVE consists of a Riemannian metric $\overline{g}$ on $\Sigma$ and a symmetric two tensor $K$ verifying the constraint equations 
\begin{align}\label{const}
\left\{
\begin{array}{ll}
\overline{\nabla}^jK_{ij}-\overline{\nabla}_i\text{tr}_{\overline{g}}K=0\\
\overline{\text{R}}-|K|^2+(\text{tr}_{\overline{g}}K)^2=0
\end{array},
\right.
\end{align}
where $\overline{\nabla},\overline{\text{R}}$ are the covariant derivative and scalar curvature intrinsic to $\overline{g}$.
\begin{figure}[h!]
  \centering
    \includegraphics[width=0.6\textwidth]{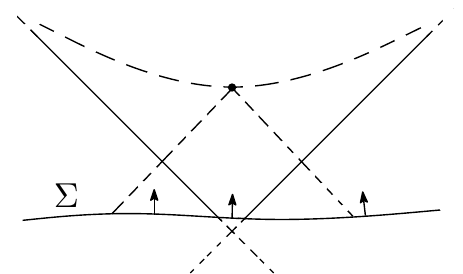}
      \caption{}
\label{SigmaIVP}
\end{figure}

The instability of the Schwarzschild singularity (w.r.t. the forward Cauchy problem) can already be seen by examining the maximal developments of initial data sets on $\Sigma$ arising from the celebrated Kerr \cite{Daf} (explicit) 2-parameter $\mathcal{K}(a,M)$ family of solutions -- of which Schwarzschild is a subfamily ($a=0$). For $a\neq0$ the singularity completely disappears and the corresponding (maximal) developments extend smoothly up to (and including) the Cauchy horizons (which lie far beyond the hypersurface $uv=1$). Moreover, taking $|a|\ll 1$, 
the `difference' of the corresponding initial data sets from the Schwarzschild one (with the same $M>0$), measured in standard Sobolev norms,\footnote{The difference can be defined, for example, component wise for the two pairs of  2-tensors with respect to a common coordinate system and measured in $W^{s,p}$ Sobolev spaces used in the literature \cite{Choq2}.} can be made arbitrarily small. 

In fact, the Schwarzschild singularity is conjecturally unstable under {\it generic} perturbations on $\Sigma$. According to a scenario proposed by Belinski$\check{\i}$-Khalatnikov-Lifshitz \cite{BKL} originally formulated for cosmological singularities, in general, one should expect solutions to exhibit oscillatory behaviour towards the singularity.\footnote{Homogeneous spacetimes with singularities that exhibit oscillatory behaviour have been rigorously studied by Ringstr$\ddot{o}$m \cite{Rin} for the Euler-Einstein system with Bianchi IX symmetry.} Although this work was not supported by much rigorous evidence, it has received a lot of attention over the years, see \cite{Rend} and the references therein (and \cite{Gar} for related numerics). On the other hand, there is a growing expectation that, at least in a neighbourhood of subextremal Kerr, the dominant scenario inside the black hole is the formation of Cauchy horizons and (weak) null singularities. This has been supported by rigorous studies on spherically symmetric charged matter models, see works by Poisson and Israel \cite{IP}, Ori \cite{Ori} and recently by Dafermos \cite{Daf}.

However, it is not inferred from the existing literature whether the non-oscillatory type of singularity observed in Scwarzschild is an isolated phenomenon for the EVE 
in some neighbourhood of the Schwarzschild initial data on 
$\Sigma$ or part of a larger family. A priori it is not clear what to expect, since one might argue that such a special singularity is a mathematical artefact due to spherical symmetry. Therefore, we pose the following question:
\begin{center}
{\it Is there a class of non-spherically symmetric Einstein vacuum spacetimes which develop a first singularity of Schwarzschild type?}
\end{center}
The goal of the present paper is to answer the preceding question in the affirmative. A Schwarzschild type singularity here has the meaning of a first singularity in the vacuum development which has the same geometric blow up profile with Schwarzschild and which can be seen by a foliation of uniformly spacelike hypersurfaces; hence, not contained in a Cauchy horizon.  We confine the question to the formation of one singular sphere in the vacuum development in the same manner as in Schwarzschild, where each point on the sphere can be understood as a distinct ideal singular point of the spacetime in the language of TIPs \cite{GKR}. Ideally one would like to study the forward problem and identify initial data for the EVE on $\Sigma$ (Figure \ref{SigmaIVP}) that lead to such singularities. Although this is a very interesting problem, we find it far beyond reach at the moment. Instead we study the existence problem backwards-in-time.
\begin{figure}[h!]
  \centering
    \includegraphics[width=0.7\textwidth]{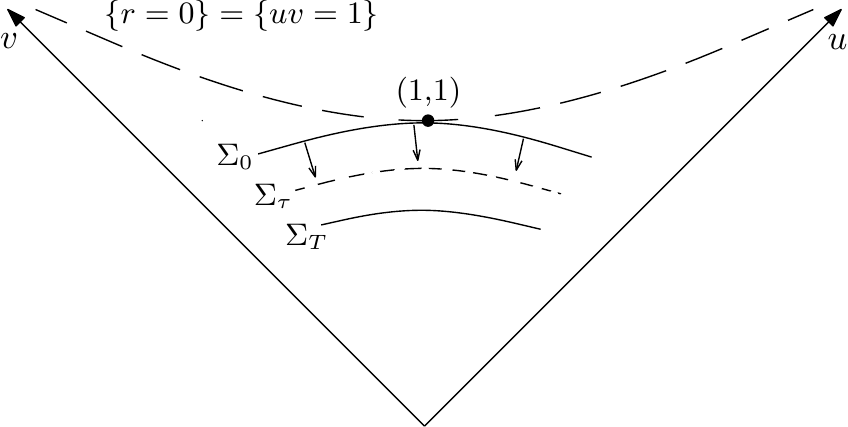}
      \caption{The black hole region in Kruskal's extension.}
\label{genS_fol}
\end{figure}

More precisely, we adopt the following plan: Let $\Sigma_0^3$ be a spacelike hypersurface in Schwarzschild, tangent\footnote{The tangency here should be understood with respect to the differential structure of the Kruskal maximal extension induced by the standard $u,v,\theta,\phi$ coordinates (\ref{Schw}).} at a single sphere of the singular
hyperbola $r=0$ inside the black hole; Figure \ref{genS_fol}. We assume, without loss of generality,\footnote{Recall that the vector field tangent to the $r=\text{const.}$ hypersurfaces (Figure \ref{Krusk}) is Killing and we may hence utilize it to shift $\Sigma_0$ and $(u=1,v=1)$ to whichever point on $\{uv=1\}$ we wish; Figure (\ref{genS_fol}).} that the tangent sphere is $(u=1,v=1$) in Kruskal coordinates (\ref{Schw}). Consider now initial data sets $(\overline{g},K)$ on $\Sigma_0$ for the EVE (\ref{EVE}), which have the same singular behaviour to leading order at $(u=1,v=1)$ with the induced Schwarzschild initial data set $({^S\overline{g}},{^SK})$ on $\Sigma_0$ and solve the EVE backwards, as depicted in the 2-dim Figure \ref{genS_fol}, without symmetry assumptions. 

Realizing the above plan we thus prove the existence of a class of non-spherically symmetric vacuum spacetimes for which $(1)$ the leading asymptotics of the blow up of curvature and in general of all the geometric quantities (metric, second fundamental form etc.) coincide with their Schwarzschild counterparts, as one approaches the singularity, and $(2)$  
the singularity is realized as the limit of uniformly spacelike hypersurfaces, which in the forward direction ``pinch off" in finite time at one sphere. Conversely, we visualise the backward evolution of $(\Sigma_0,\overline{g},K)$ in the following manner: At `time' $\tau=0$ the initial slice $\Sigma_0$ is a two ended spacelike (3-dim) hypersurface with a sphere singularity at
$(u=1,v=1)$. Once $\Sigma_0$ evolves through (\ref{EVE}),
it becomes instantaneously a smooth spacelike hypersurface $\Sigma_\tau,\tau>0$ and the
singular {\it pinch} opens up; Figure \ref{S_fol3D}.
\begin{figure}[h!]
  \centering
    \includegraphics[width=0.6\textwidth]{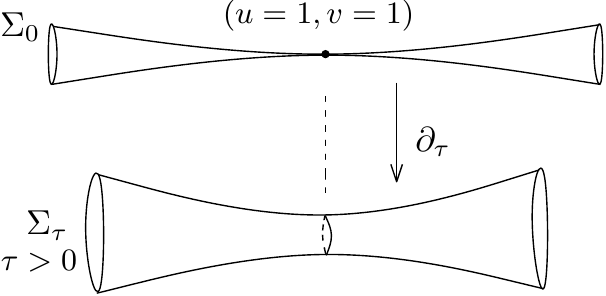}
      \caption{}
\label{S_fol3D}
\end{figure}

The main difficulty to overcome in the backward local existence problem is the essential singularity on $\Sigma_0$, which of course renders it beyond the scope of the classical local existence theorem for the Einstein equations \cite{Choq}, even its latest state of the art improvement by Klainerman-Rodnianski-Szeftel \cite{KRS}, which requires at the very least the curvature of the initial hypersurface to be in $L^2$. For the Schwarzschild initial data set $(^S\overline{g},{^SK})$ on $\Sigma_0$, and hence for perturbed initial data sets $(\overline{g},K)$ with the same leading order geometry at $(u=1,v=1)$, it is not hard to check (\S\ref{Schwarz}) that the initial curvature is at the singular level 
\begin{align}\label{RnotinL2}
\overline{\text{R}}\not\in L^p(\Sigma_0),\qquad\nabla K\not\in L^p(\Sigma_0)&&p\ge\frac{5}{4}.
\end{align}
To our knowledge, the only general local existence results, without symmetry assumptions, for the EVE (\ref{EVE}) with singular initial curvature not in $L^2$ have been achieved only fairly recently by Luk-Rodnianski \cite{LukRod1,LukRod2} and Luk \cite{Luk} for the characteristic initial value problem. 
However, the singular level of the Schwarzschild singularity  is much stronger than the ones they consider; delta curvature singularities \cite{LukRod1,LukRod2}
and weak null singularities \cite{Luk}. Also, since we are interested 
in the spacelike character of the singularity of the perturbed spacetimes we wish to construct, the non-characteristic initial value problem seems to be a more natural setting to consider.

We proceed now to formulate a first version of our main results; for more precise statements, in terms of weighted Sobolev spaces, see Theorems \ref{mainthm}, \ref{mainthm2}, \ref{conststab}.  
\begin{theorem}\label{thmA}
There exists $\alpha>0$ sufficiently large, such that for every triplet
$(\Sigma_0,\overline{g},K)$ verifying:\\
$(i)$ the constraints (\ref{const}),\\
$(ii)$ $\overline{g}={^S\overline{g}}+r^\alpha\mathcal{O}$, $K={^SK}+r^\alpha u$, where $\mathcal{O},u$ are 2-tensors on $\Sigma_0$ bounded in $H^4,H^3$ respectively,\\
$(iii)$ $\|\overline{g}-{^S\overline{g}}\|_{L^\infty(\Sigma_0)}\ll1$,\\
there exists a $H^4$ local solution $g$ to the Einstein vacuum equations (\ref{EVE}) with initial data $(\overline{g},K)$, unique up to isometry, in the backward region to $\Sigma_0$, foliated by $\{\Sigma_\tau\}_{\tau\in[0,T]}$ (Figure \ref{genS_fol}); the time of existence $T>0$ depends continuously on the norms of $\mathcal{O},u$ 
and the exponent $\alpha>0$.
\end{theorem}
The fact that non-trivial initial data sets in compliance with Theorem \ref{thmA} exist is not at all obvious nor standard. It has to be shown essentially that for any large parameter $\alpha>0$, the set of non-spherically symmetric solutions to the constraint equations (\ref{const}), having the asymptotics $(ii)$, is non-empty. This is achieved in the following theorem.
\begin{theorem}\label{thmB}
Let $\alpha>0$ be sufficiently large, consistent with Theorem \ref{thmA}. There exists a small ball $\mathcal{B}\subset H^4(\partial\Sigma_0)\times H^3(\partial\Sigma_0)$, such that for every pair $(\mathcal{O},u)$ of 2-tensors defined on the boundary of $\Sigma_0$, which lie inside that ball, there exists at least one extension of $(\mathcal{O},u)$ as 2-tensors in $\Sigma_0$ satisfying the assumptions of Theorem \ref{thmA}. In particular, the pair $(\overline{g},K)$ defined via $(ii)$ in Theorem \ref{thmA} satisfies the constraints (\ref{const}).
\end{theorem}
Finally, let us emphasize the fact that the above spacetimes are very special in that they agree with Schwarzschild at the singularity to a high (but finite) order -- this is captured by the large exponent $\alpha>0$ in Theorem \ref{thmA} -- and therefore are {\it non-generic}.
The need to choose $\alpha$ large may be seen however natural to some extend in view of the instability of the Schwarzschild singularity; from the point of view of the forwards-in-time problem. Indeed, the stable perturbations of the Schwarzschild singularity must form a strict subclass of all perturbed vacuum developments.

\subsection{Method of proof and outline}\label{Meth}

The largest part of the paper is concerned with the evolutionary part of the problem, i.e., proving Theorem \ref{thmA}. Due to the singular nature of backward existence problem described above, Figure \ref{genS_fol}, the choice of framework must be carefully considered. The standard wave coordinates approach \cite{Choq} does not seem to be feasible in our situation; one expects that coordinates would be highly degenerate at the singularity. Also, the widely used CMC gauge condition is not applicable, since the mean curvature of the initial hypersurface $\Sigma_0$ blows up (\S\ref{Schwarz}). Instead, we find it more suitable to use orthonormal frames and rewrite the EVE one order higher as a quasilinear Yang-Mills hyperbolic system of equations \cite{Mon2,KRS}, under a Lorentz gauge condition,\footnote{The analogue of a wave gauge for orthonormal frames.} for the corresponding connection 1-forms. We recall briefly this framework in \S\ref{DEE}.

However, even after expressing the EVE in the above framework, the singular level of initial configurations do not permit a direct energy estimate approach. In addition to (\ref{RnotinL2}), one can see (\S\ref{Schwarz}) that neither is the second fundamental form in $L^2$
\begin{align}\label{KnotinL2}
K\not\in L^2(\Sigma_0).
\end{align}
Note that the latter is at the level of one derivative in the metric. Hence, near the singularity the perturbed spacetimes we wish to construct do not even make sense as weak solutions of the EVE (\ref{EVE}).
Therefore, it is crucial that we use the background
Schwarzschild spacetime to recast the evolution equations in a new form having more regular initial data. We do this in \S\ref{stab} by considering a new system of equations for the `difference' between the putative perturbed spacetime and Schwarzschild. The resulting equations have now regular initial data and they are eligible for an energy method, but there is a price to pay. The coefficients of the new system will depend on the Schwarzschild geometry and will necessarily be highly singular at $r=0$. We compute in \S\ref{Schwarz} the precise blow up orders of the Schwarzschild connection coefficients, curvature etc. Nevertheless, the issue of evolving singular initial data has become the more tractable problem of finding appropriate weighted solution spaces for the final singular equations. 

In \S\ref{wHs} we introduce the weighted Sobolev spaces which yield the desired flexibility in proving energy estimates. The right weights are given naturally by the singularities in the coefficients of the resulting equations, namely, powers of the Schwarzschild radius function $r$ with a certain analogy corresponding to the order of each term. After stating the general local existence theorems in \S\ref{main} and a more precise version of Theorem \ref{thmA}, we proceed to its proof via a contraction mapping argument which occupies Section \ref{fixedpt}. Therein we derive the main weighted energy estimates by exploiting the asymptotic analysis at $r=0$ of the Schwarzschild components (\S\ref{Schwarz}).
It is necessary in our result that the power of $r$, $\alpha>0$, in the weighted norms is sufficiently large; cf. assumption $(ii)$ in Theorem \ref{thmA}. In the estimating process certain {\it critical} terms are inevitably generated, because of the singularities in the coefficients of the system we are working with; these terms are critical in that they appear with larger weights than the ones in the energy we are trying to control and thus prevent the estimates from closing. The exponent $\alpha>0$ is then picked sufficiently large such that these critical terms have an overall favourable sign; this allows us to drop the critical terms and close the estimates. 

The largeness of $\alpha$ forces the perturbed spacetime to agree asymptotically with Schwarzschild to a high order at the singularity. Although the latter may seem restrictive, it is quite surprising to us that there even exists a suitable choice of $\alpha$ which makes the argument work in the first place. A closer inspection of our method reveals that it is very sensitive with respect to certain asymptotics of the coefficients in the equations that happen to be just borderline to allow an energy-based argument to close. The most important of these are the blow up order of the sectional curvature (\ref{K}) and the rate of growth of the Schwarzschild radius function $r$ backwards in time. The latter corresponds to the `opening up' rate of the {\it neck pinch} of the singular initial hypersurface $\Sigma_0$, Figure \ref{S_fol3D}. In this sense the Schwarzschild singularity is exactly at the threshold that our energy-based method can tolerate. 

In the last section, \S\ref{singconst}, we study the constraint equations (\ref{const}) in a perturbative manner about the Schwarzschild singular initial data set $(^S\overline{g},{^SK})$ on $\Sigma_0$. First we prescribe boundary data\footnote{The boundary of $\Sigma_0$ consists of two components diffeomorphic to $\mathbb{S}^2$.} at $\partial\Sigma_0$ close to Schwarzschild (in usual Sobolev norms, see Theorem \ref{thmB}) and then look for solutions $(\overline{g},K)$ to constraints inside the class of singular initial data sets that are compatible with our assumptions on the evolutionary problem; cf. Theorem \ref{thmA}. It is very likely that the {\it generic}  behaviour at $(u=1,v=1)$ of solutions to (\ref{const}) with this perturbed data at the boundary of $\Sigma_0$ will not agree with Schwarzschild to a high enough order for Theorem \ref{thmA} to be applicable. However, we prove the existence of solutions for which the asymptotic behaviour required by Theorem \ref{thmA} indeed holds.
We obtain Theorem \ref{thmB} as an application of the implicit function theorem. The proof is based on properties of the constraint map \cite{Mon1}, defined in the weighted Sobolev spaces introduced above. In order to prove the surjectivity of the linearized constraint map around Schwarzschild, we use an indirect argument from \cite{CorSch} adapted to our context. One of the main ingredients in the proof involves the Fredholm property of a linearized operator, see Proposition \ref{Fred}, which we prove in \S\ref{confFred} by deriving certain weighted elliptic estimates.

\subsection{Final Comments; Possible applications}\label{appl}

To our surprise the present evolution bears some resemblance at an analytical level with a prior work on the stability of singular Ricci solitons \cite{ACF}. Although of different nature, hyperbolic/parabolic (respectively), they share a couple of key features such as the opening up rate of the singularity and the ``borderline'' singularities in the coefficients involved.  

The understanding of the question of stability of singularities in Einstein's equations and the behaviour of solutions near them is of great significance in the field. However, in general very little is known. In terms of rigorous results, substantial progress has been made under symmetry assumptions in the presence of matter \cite{Chr1} \cite{Chr2} \cite{Rin} \cite{Daf}. Moreover, certain matter models enjoy the presence of a monotonic quantity, which has been employed to study the stability of singularity formation in the general non-symmetric regime, cf. recent work of Rodnianski-Speck \cite{RodSp} on the FLRW big bang singularity. 
This is in contrast with the vacuum case and the unstable nature of the Schwarzschild singularity. We emphasize again the fact that the method developed herein does not impose any symmetry assumptions nor it relies on any monotonicity. It should be noted as well that it does not depend on whether the particular singularity type is {\it generic} or not. On the contrary, we hope that it can serve as a robust new method which can be generally employed to produce classes of examples of non-symmetric singular solutions to the Einstein field equations, which until now are only known to exist under special symmetry assumptions and for which the general stability question may be out of reach.

After our treatment of singular initial data containing a single sphere of $\{uv=1\}$, a reasonable next step would be to study whether the construction of non-spherically symmetric vacuum spacetimes containing an arc of the singular hyperbola (Figure \ref{genS_fol}) is possible or even the whole singularity $r=0$. Certainly this is a more restrictive question and at first glance not so obvious how to formulate it as a backward initial value problem problem for the EVE. However, we hope that the method developed herein could help approach this direction.

Lastly, one could try to perform a global instead of a local construction by considering a Cauchy hypersurface $\Sigma_0$ extending to spacelike infinity. We expect this follows readily from the work here, but we do not pursue it further. Perhaps a gluing construction could also be achieved.

\subsection{Acknowledgements}

I would like to thank my advisor Spyros Alexakis for his valuable suggestions and interest in this work. I am grateful to Mihalis Dafermos for helpful comments on a preliminary version of the present paper.
Thanks also go to Yannis Angelopoulos, Jesse Gell-Redman, Niky Kamran, Volker Schlue, Arick Shao for useful discussions. The author is supported by an Onassis Foundation Fellowship.

\section{The Einstein equations as a quasilinear Yang-Mills system}\label{DEE}

The Einstein vacuum equations (\ref{EVE}), by virtue of the second Bianchi identity, imply the vanishing of the divergence of the Riemann curvature tensor. 
Decomposing the latter with respect to an orthonormal frame, which satisfies a suitable gauge condition, it results to a quasilinear second order hyperbolic system of equations for the connection 1-forms corresponding to that frame, which bears resemblance to the semilinear Yang-Mills \cite{Mon2}. Recently this formulation of the EVE played a key role in the resolution of the bounded $L^2$ curvature conjecture \cite{KRS}. 
In this section we express the EVE (\ref{EVE}) in the above setting, which we are going to use to directly solve the Cauchy problem. This necessitates some technical details which are carried out in Appendix \ref{app1}. Also, to avoid additional computations we write all equations directly in scalar non-tensorial form.\footnote{It will be clear though which are the covariant expressions; see also \cite{KRS}.} 

\subsection{Cartan formalism}\label{Cart}

Let $(\mathcal{M}^{1+3},g)$ be a Lorenzian manifold and let $\{e_0,e_1,$ 
$e_2,e_3\}$
be an orthonormal frame; $g_{ab}:=m_{ab}=\text{diag}(-1,1,1,1)$. Assume also that $\mathcal{M}^{1+3}$ has the differential structure of $\Sigma\times[0,T]$, where each leaf $\Sigma\times\{\tau\}=:\Sigma_\tau$ is a 3-dim spacelike hypersurface. 
We denote the connection 1-forms associated to the preceding frame by
\begin{align}\label{A}
(A_X)_{ij}:=g(\nabla_Xe_i,e_j)=-(A_X)_{ji},
\end{align}
where $\nabla$ is the $g$-compatible connection of $\mathcal{M}^{1+3}$.
Recall the definition of the Riemann curvature tensor
\begin{align}\label{Riem}
R_{\mu\nu ij}:=g(\nabla_{e_\mu}\nabla_{e_\nu}e_i-\nabla_{e_\nu}\nabla_{e_\mu}e_i,e_j).
\end{align}
By the former definition of connection 1-forms, using $m_{ab}$ to raise and lower indices, we write
\begin{align*}
\nabla_{e_a}e_b={(A_a)_b}^ke_k.
\end{align*}
Hence, we have
\begin{align*}
\nabla_{e_\mu}\nabla_{e_\nu} e_i=&\;\nabla_{e_\mu}(\nabla_{e_\nu}e_i)-
\nabla_{\nabla_{e_\mu}e_\nu}e_i=\nabla_{e_\mu}\big({(A_\nu)_i}^ke_k\big)
-{(A_\mu)_\nu}^k{(A_k)_i}^ce_c\\
=&\;e_\mu{(A_\nu)_i}^ke_k+
{(A_\nu)_i}^k{(A_\mu)_k}^de_d-{(A_\mu)_\nu}^k{(A_k)_i}^ce_c
\end{align*}
Therefore, we get the following expression for the components of the Riemann curvature
\begin{align}\label{Riem2}
R_{\mu\nu ij}=&\;e_\mu(A_\nu)_{ij}
-e_\nu(A_\mu)_{ij}+
{(A_\nu)_i}^k(A_\mu)_{kj}-{(A_\mu)_i}^k(A_\nu)_{kj}\\
\notag&-{(A_\mu)_\nu}^k(A_k)_{ij}+{(A_\nu)_\mu}^k(A_k)_{ij}
\end{align}
or setting
\begin{align}\label{[Amu,Anu]}
([A_\mu,A_\nu])_{ij}=
{(A_\mu)_i}^k(A_\nu)_{kj}-{(A_\nu)_i}^k(A_\mu)_{kj}
\end{align}
we rewrite
\begin{align}\label{Fmunu}
(F_{\mu\nu})_{ij}:=R_{\mu\nu ij}=&\;e_\mu(A_\nu)_{ij}
-e_\nu(A_\mu)_{ij}-([A_\mu,A_\nu])_{ij}-{(A_{[\mu})_{\nu]}}^k(A_k)_{ij},
\end{align}
where by standard convention
\begin{align*}
{(A_{[\mu})_{\nu]}}^k(A_k)_{ij}:=
{(A_\mu)_\nu}^k(A_k)_{ij}-{(A_\nu)_\mu}^k(A_k)_{ij}.
\end{align*}
In the same manner we compute the covariant derivative of the Riemann tensor:
\begin{align}\label{DRiem}
\notag \nabla_\sigma R_{\mu\nu ij}=&\;e_\sigma(F_{\mu\nu})_{ij}
-{(A_\sigma)_\mu}^k(F_{k\nu})_{ij}
-{(A_\sigma)_\nu}^k(F_{\mu k})_{ij}\\
&\notag-{(A_\sigma)_i}^k(F_{\mu\nu})_{kj}
-{(A_\sigma)_j}^k(F_{\mu\nu})_{ik}\\
=&\;e_\sigma(F_{\mu\nu})_{ij}
-{(A_\sigma)^k}_{[\mu}(F_{\nu] k})_{ij}
-([A_\sigma,F_{\mu\nu}])_{ij}
\end{align}

Recall the transformation law of the above quantities under change of frames:
Let $\{\tilde{e}_i\}_0^3$ be an orthonormal frame on $\mathcal{M}^{1+3}$ such that
\begin{align}\label{O}
\tilde{e}_a=O^k_a e_k
\end{align}
and let $(\tilde{A}_X)_{ij}:=g(\nabla_{X}\tilde{e}_i,\tilde{e}_j)$ be the corresponding
connection 1-forms. Then
\begin{align}\label{trans}
(\tilde{A}_X)_{ij}=O^b_iO^c_j(A_X)_{bc}+X(O^b_i)O^c_jm_{bc}.
\end{align}
In addition, from (\ref{O}) we have
\begin{align*}
\nabla_X\tilde{e}_a=&\;X(O^k_a)e_k+O^k_a\nabla_Xe_k\\
{(\tilde{A}_X)_{a}}^d\tilde{e}_d=&\;X(O^k_a)e_k+O^k_a{(A_X)_k}^de_d
\end{align*}
or
\begin{align}\label{X(O)}
X(O^l_a)={(\tilde{A}_X)_{a}}^dO^l_d-O^k_a{(A_X)_k}^l.
\end{align}
\subsection{$\nabla\times\text{Ric}=0$}

Now we proceed by assuming that the {\it curl} of the Ricci tensor of the metric $g$ vanishes:
\begin{align}\label{curlRic}
\nabla_iR_{\nu j}- \nabla_jR_{\nu i}=0,
\end{align}
where $R_{ab}:={R_{\mu ab}}^\mu$.
A direct implication of the (contracted) second Bianchi identity is
that the divergence of the Riemann curvature tensor satisfies
\begin{align}\label{divRiem}
\nabla^\mu R_{ij\nu \mu}= \nabla_iR_{\nu j}- \nabla_jR_{\nu i}=0.
\end{align}
Thus, it follows from (\ref{DRiem}) that
\begin{align}\label{divRiem2}
e^\mu(F_{\mu\nu})_{ij}
-{(A^\mu)^k}_{[\mu}(F_{\nu] k})_{ij}
-([A^\mu,F_{\mu\nu}])_{ij}=0
\end{align}
or by (\ref{Fmunu})
\begin{align}\label{boxA}
&\boxdot(A_\nu)_{ij}
-e^\mu e_\nu(A_\mu)_{ij}-e^\mu([A_\mu,A_\nu])_{ij}-e^\mu\big({(A_{[\mu})_{\nu]}}^k(A_k)_{ij}\big)\\
\notag=&\;{(A^\mu)^k}_{[\mu}(F_{\nu] k})_{ij}
+([A^\mu,F_{\mu\nu}])_{ij},
\end{align}
where $\boxdot:=-e_0^2+e_1^2+e^2_2+e^2_3$ is the non-covariant box with respect to the frame $e_i$. Since
\begin{align*}
[e_\mu,e_\nu]= \nabla_{e_\mu} e_\nu- \nabla_{e_\nu} e_\mu={(A_{[\mu})_{\nu]}}^ke_k,
\end{align*}
(\ref{boxA}) takes the equivalent form
\begin{align}\label{boxA2}
\notag\boxdot(A_\nu)_{ij}
-e_\nu e^\mu(A_\mu)_{ij}=&\;{(A^{[\mu})_{\nu]}}^ke_k(A_\mu)_{ij}
+e^\mu([A_\mu,A_\nu])_{ij}+e^\mu\big({(A_{[\mu})_{\nu]}}^k(A_k)_{ij}\big)\\
&+{(A^\mu)^k}_{[\mu}(F_{\nu] k})_{ij}
+([A^\mu,F_{\mu\nu}])_{ij},
\end{align}
$\nu,i,j=0,1,2,3$. 
We remark that (\ref{boxA2}) is an equation of scalar functions.

\subsection{Choice of gauge}

Note that the preceding equation is not of hyperbolic type. We convert (\ref{boxA2}) into a quasilinear hyperbolic system of equations
by imposing a Lorentz gauge condition on the orthonormal frame $\{e_i\}^3_0$:\footnote{A wave type gauge essentially for $e_i$. The Coulomb gauge is another alternative which is used in \cite{KRS}. We do not employ it here.}
\begin{align}\label{gauge}
A^2=(\text{div}A_X)_{ij}:= \nabla^\mu(A_\mu)_{ij}-(A_{ \nabla_{e^\mu} e_\mu})_{ij}=
e^\mu(A_\mu)_{ij}-{(A^\mu)_\mu}^k(A_k)_{ij},
\end{align}
where by $A^2$ we denote some quadratic expression in the connection 
coefficients $(A_\nu)_{ij}$ varying in $ij$. This a freedom one has in choosing the frame $e_i$; see Lemma \ref{gaugemap}. Under (\ref{gauge}), the equation (\ref{boxA2}) becomes the quasilinear second order 
\begin{align}\label{boxA3}
\notag\boxdot(A_\nu)_{ij}=&\;{(A^{[\mu})_{\nu]}}^ke_k(A_\mu)_{ij}
+e^\mu([A_\mu,A_\nu])_{ij}
+e^\mu\big({(A_{[\mu})_{\nu]}}^k(A_k)_{ij}\big)\\
&+{(A^\mu)^k}_{[\mu}(F_{\nu] k})_{ij}
+([A^\mu,F_{\mu\nu}])_{ij}+e_\nu(A^2)+e_\nu\big({(A^\mu)_\mu}^k(A_k)_{ij}\big)
\end{align}
\subsection{The reduced equations; Initial data for EVE}\label{ID}

Following (\ref{divRiem})-(\ref{boxA3}) we actually see that the equation
\begin{align}\label{redeq}
 \nabla_i\text{R}_{\nu j}- \nabla_j\text{R}_{\nu i}+&e_\nu\big(\text{div}A_X- A^2\big)_{ij}\\
\notag=:&\;H_{\nu ij}=\text{(LHS of (\ref{boxA3}))}-\text{(RHS of (\ref{boxA3}))},
\end{align}
holds true for every Lorentzian metric $g$ and orthonormal frame $\{e_i\}^3_0$,
without any additional assumptions or gauge condition. We call
 $H_{\nu ij}=0$, i.e., the system (\ref{boxA3}), the reduced equations. 
We note that even after the gauge fixing, the reduced equations are not equivalent to the EVE (\ref{EVE}), but only imply the vanishing of the {\it curl} of the Ricci tensor (\ref{curlRic}). However, one may suitably prescribe initial data for (\ref{boxA3}) such that they lead to solutions of the EVE and which are consistent with the Lorentz gauge condition (\ref{gauge}). 

Now we address the initial value problem for the reduced equations $H_{\nu ij}=0$ aiming to the EVE. To solve the equation (\ref{boxA3}) one needs an equation relating the evolution of the orthormal frame $\{e_i\}^3_0$ to that of the connection 1-forms. Let $\partial_0,\partial_1,\partial_2,\partial_3$ be a reference frame\footnote{Not orthonormal or coordinates, simply a basis frame.} in $\Sigma\times[0,T]$ ($\partial_0$ transversal direction). We express $e_i$ in terms of $\partial_a$:
\begin{align}\label{epartial}
e_i=O^a_i\partial_a
\end{align}
By virtue of the diffeomorphism invariance of the EVE, we may assume that the timelike unit vector of the orthonormal frame $\{e_i\}^3_0$ of the spacetime we solve for is $e_0=\partial_0$. Doing so we deduce
\begin{align*}
\partial_0(O^a_i)=\mathcal{L}_{e_0}\big(\hat{\partial}_a(e_i)\big)=
\mathcal{L}_{\partial_0}(\hat{\partial}_a)e_i+\hat{\partial}_a([\partial_0,e_i])=
O^b_i\mathcal{L}_{\partial_0}(\hat{\partial}_a)\partial_b+\hat{\partial}_a([e_0,e_i]),
\end{align*}
where $\mathcal{L}$ denotes the Lie derivative and $\hat{\partial}_a$ is the 1-form dual to $\partial_a$. Setting $[\partial_0,\partial_b]=:\Gamma^c_{[0b]}\partial_c$ we rewrite
\begin{align}\label{partial_0frame}
\partial_0(O^a_i)=-O^b_i\Gamma^a_{[0b]}+{(A_{[0})_{i]}}^kO^a_k.
\end{align}
Now we proceed to formulate the necessary and sufficient conditions on the initial data set of the reduced equations (\ref{boxA3}), coupled to (\ref{partial_0frame}), such that the corresponding solution yields a solution to the EVE. The following proposition is proved in \S\ref{app1.1}. 
\begin{proposition}\label{IDprop}
Let $(A_\nu)_{ij},O^a_i$ be a solution of (\ref{boxA3}),(\ref{partial_0frame}), arising from initial configurations subject to
\begin{align}\label{initAO}
(A_\nu)_{ij}(\tau=0)=-(A_\nu)_{ji}(\tau=0)&&\partial_0(A_\nu)_{ij}(\tau=0)=-\partial_0(A_\nu)_{ji}(\tau=0)\\
\notag O^a_0(\tau=0)={I_0}^a
\end{align}
and 
\begin{align}\label{initA}
(\text{div}A_X)_{ij}-A^2=0\;\;\Longleftrightarrow\;\;e^\mu(A_\mu)_{ij}-{(A^\mu)_\mu}^k(A_k)_{ij}-A^2=0\\
\notag\text{Ric}_{ab}(g)=0\;\;\Longleftrightarrow\;\;e^\mu(A_\nu)_{i\mu}-e_\nu(A^\mu)_{i\mu}-([A^\mu,A_\nu])_{i\mu}-{(A^{[\mu})_{\nu]}}^k(A_k)_{i\mu}=0
\end{align}
on $\Sigma_0$. 
Then the latter solution corresponds to an Einstein vacuum spacetime $(\mathcal{M}^{1+3},g)$ and furthermore the frame $\{e_i\}^3_0$ (\ref{epartial}) is $g$-orthonormal, $e_0=\partial_0$, and satisfies the Lorentz gauge condition (\ref{gauge}).
\end{proposition}
\begin{remark}
Note that the second part of (\ref{initA}) includes the constraints (\ref{const}); $\text{R}_{0b}=\text{R}_{00}-\frac{1}{2}\text{R}=0$, $b=1,2,3$, on $\Sigma_0$. 
Once we have chosen the orthonormal frame initially and the initial data components $(A_0)_{ij}(\tau=0)$, which correspond to the $\partial_0$ derivative of $\{e_i\}^3_0$, then the rest of the initial data set of (\ref{boxA3}) is fixed by the condition (\ref{initA}), i.e., the Lorentz gauge and the EVE on the initial hypersurface $\Sigma_0$; see Lemma \ref{gaugemap}, Remark \ref{initfixed}.
\end{remark}

\section{The Schwarzschild components}\label{Schwarz}

We fix an explicit Schwarzschild orthonormal reference frame and compute the corresponding connection coefficients, which we then use to find the leading asymptotics of the second fundamental form and curvature of the initial singular hypersurface $\Sigma_0$ in Schwarzschild. Knowing the precise leading blow up behaviour 
of these quantities is crucial for the study of local well-posedness in the next section. For distinction, we denote Schwarzschild components with an upper left script $^S$.

Let us consider a specific foliation of spacelike hypersurfaces $\Sigma_\tau$, $\tau\in[0,T]$, for the backward problem in a neighbourhood 
of $(u=1,v=1)$; Figure \ref{genS_fol}. 
For convenience\footnote{It is easy to see that the following leading asymptotics we derive are independent of the particular choice of foliation.} let
\begin{align}\label{fol}
\Sigma_\tau:\qquad-\frac{1}{2}(u+v)+1=\tau&&(u,v)\in(1-\epsilon,1+\epsilon)^2,\;\tau\in[0,T].
\end{align}
\begin{figure}[h!]
  \centering
    \includegraphics[width=0.7\textwidth]{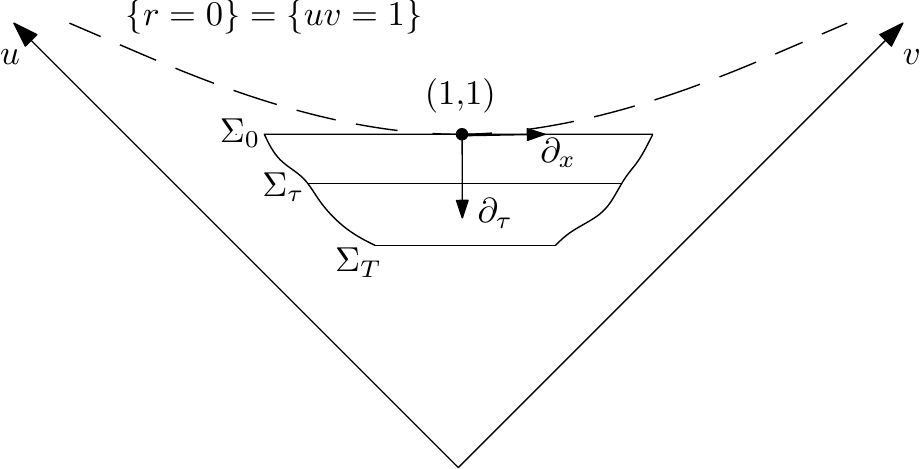}
      \caption{The foliation (\ref{fol}) in the interior of the black hole.}
\label{S_fol}
\end{figure}

In temporal and spatial coordinates $\tau,x$
\begin{align}\label{dtau,dx}
\partial_\tau:=-\partial_u-\partial_v,&&\partial_x:=\partial_u-\partial_v\\
\notag&&x=\frac{1}{2}(u-v),
\end{align}
the metric (\ref{Schw}) takes the form
\begin{align}\label{admetric}
{^Sg}=-\Omega^2d\tau^2+\Omega^2dx^2+r^2(d\theta^2+\sin^2\theta d\phi^2),&&\Omega^2=\frac{32M^3}{r}e^{-\frac{r}{2M}}.
\end{align}
By (\ref{uv}),(\ref{dtau,dx}) $r$ is related to $\tau,x$ via
\begin{align}
(1-\tau)^2-x^2=(1-\frac{r}{2M})e^{\frac{r}{2M}},
\end{align}
from which one can derive the following formulas:
\begin{align}\label{dr/dtau}
\partial_\tau r=\frac{\Omega^2}{4M}(1-\tau),
\qquad\partial_xr=\frac{\Omega^2}{4M}x\\
\notag\partial_\tau\Omega^2=-\frac{\Omega^4}{4M}(\frac{1}{r}+\frac{1}{2M})(1-\tau),&\qquad\partial_x\Omega^2=
-\frac{\Omega^4}{4M}(\frac{1}{r}+\frac{1}{2M})x
\end{align}
\begin{remark}\label{r2}
The above first two identities yield the leading asymptotics:
\begin{align}\label{rasym}
r^2\sim 16M^2(\frac{x^2}{2}+\tau),&&\text{as $\tau,x\rightarrow0$}.
\end{align}
\end{remark}
Directly from the form of the induced metric on $\Sigma_\tau$,
\begin{align}\label{gbar}
{^S\overline{g}}=\Omega^2dx^2+r^2(d\theta^2+\sin^2\theta d\phi^2),
\end{align}
we compute the corresponding induced volume form 
\begin{align}\label{dmu}
d\mu_{{^S\overline{g}}}=\Omega r^2\sin\theta dxd\theta d\phi=\big[4\sqrt{2}M^{\frac{3}{2}}r^\frac{3}{2}+O(r^2)\big]\sin\theta dxd\theta d\phi
\end{align}
and its rate of change along $\partial_\tau$ using (\ref{dr/dtau}):
\begin{align}\label{dmu/dtau}
\partial_\tau d\mu_{{^S\overline{g}}}=\big[\frac{12M^2}{r^2}(1-\tau)+O(\frac{1}{r})\big]d\mu_{^S\overline{g}}.
\end{align}
Normalizing, we define the Schwarzschild orthonormal frame
\begin{align}\label{Sframe}
\partial_0=\frac{1}{\Omega}\frac{\partial}{\partial \tau}&&
\partial_1=\frac{1}{\Omega}\frac{\partial}{\partial x}&&
\partial_2=\frac{1}{r}\frac{\partial}{\partial\theta}&&
\partial_3=\frac{1}{r\sin\theta}\frac{\partial}{\partial\phi}
\end{align}
and the relative connection coefficients  ${^S(A_\mu)_{ij}}={^Sg}(^S\nabla_{\partial_\mu}\partial_i,\partial_j)$ associated to it.
A tedious computation\footnote{One may calculate the connection coefficients using the Koszul formula
\begin{align*}
{^S(A_\mu)_{ij}}=\frac{1}{2}\bigg[{^Sg}([\partial_\mu,\partial_i],\partial_j)-{^Sg}([\partial_i,\partial_j],\partial_\mu)+{^Sg}([\partial_j,\partial_\mu],\partial_i)
\bigg]
\end{align*}
} shows that the non-zero components read
\begin{align}\label{A_mu_ij}
\notag&{^S(A_0)_{01}}=-\frac{\Omega}{8M}(\frac{1}{r}+\frac{1}{2M})x\\
\notag&{^S(A_1)_{01}}=-\frac{\Omega}{8M}(\frac{1}{r}+\frac{1}{2M})(1-\tau)\\
&{^S(A_2)_{02}}={^S(A_3)_{03}}=\frac{\Omega}{4M}\frac{1-\tau}{r}\\
\notag&{^S(A_2)_{12}}={^S(A_3)_{13}}=\frac{\Omega}{4M}\frac{x}{r}\\
\notag&{^S(A_3)_{23}}=\frac{\cot\theta}{r}
\end{align}
Recall the (spacetime) divergence formula of the connection 1-forms $X\to {^S(A_X)_{ij}}$
\begin{align}\label{divA}
{^S(\text{div}A_X)_{ij}}:=\partial^\mu{^S(A_\mu)_{ij}}-{^S(A_{ \nabla_{\partial^\mu}\partial_\mu})_{ij}}=
\partial^\mu{^S(A_\mu)_{ij}}-{^S{(A^\mu)_\mu}^b}{^S(A_b)_{ij}}
\end{align}
Utilizing (\ref{dr/dtau}) and (\ref{A_mu_ij}), we check that the first order term in the RHS of (\ref{divA}) vanishes
\begin{align}\label{divdot}
\partial^\mu{^S(A_\mu)_{ij}}=0,
\end{align}
leaving
\begin{align}\label{divA2}
{^S(\text{div}A_X)_{ij}}={^S(A_3)_{23}}{^S(A_2)_{ij}}.
\end{align}
\begin{remark}
Thus, the orthonormal frame (\ref{Sframe}) satisfies a Lorentz gauge type condition (\ref{gauge}).
\end{remark}
\begin{remark}\label{overasym}
Summarizing the above identities and formulas we obtain the following leading asymptotics at $r=0$:
\begin{align}\label{overasym1}
\notag\partial_0\sim\frac{1}{4\sqrt{2}M^\frac{3}{2}}r^\frac{1}{2}\partial_\tau&&\partial_1\sim\frac{1}{4\sqrt{2}M^\frac{3}{2}}r^\frac{1}{2}\partial_x\\
|{^SA}|\leq\frac{C}{r^{\frac{3}{2}}}&&|\partial^{(k)}{^SA}|\leq\frac{C}{r^{(k+1)\frac{3}{2}}},
\end{align}
where $C$ depends on $M>0$ and $k$. Notice that the latter asymptotics are sharp for $k=0$ and when $\partial^{(k)}=\partial_0^{(k)}$. In fact, the components of the second fundamental form of the slices ${^SK}_{ii}={^S(A_i)_{0i}}$, $i=1,2,3$, are exactly at this level. In more geometric terms we have (up to constants)
\begin{align}\label{overasym2}
|{^SK}|\sim \frac{1}{r^\frac{3}{2}}&&|\text{tr}_{^S\overline{g}}{^SK}|\sim\frac{1}{r^\frac{3}{2}}&&|^S\overline{\text{R}}|\sim\frac{1}{r^2}.
\end{align}
Thus, employing (\ref{dmu}),(\ref{rasym}) for $\tau=0$, we see that both the scalar curvature and the second fundamental of the initial singular hypersurface $\Sigma_0$ are far from being square integrable
\begin{align}\label{notL2}
\notag\int_{\Sigma_0}|{^SK}|^2d\mu_{^S\overline{g}}\sim\int^\epsilon_0\frac{1}{x^3}x^\frac{3}{2}dx=\int^\varepsilon_0\frac{1}{x^\frac{3}{2}}dx=+\infty\\
\int_{\Sigma_0}|{^S\overline{\text{R}}}|^2d\mu_{^S\overline{g}}\sim\int^\epsilon_0\frac{1}{x^4}x^{\frac{3}{2}}dx=\int^\epsilon_0\frac{1}{x^\frac{5}{2}}dx=+\infty
\end{align}
The same holds for the mean curvature of $\Sigma_0$. In fact, a similar calculation shows $\text{tr}_{^S\overline{g}}{^SK}\not\in L^p$, $p\ge\frac{5}{3}$.
\end{remark}
\section{The local-in time well-posedness}\label{stab}
\subsection{Perturbed spacetime; A transformed system}\label{pertsp}

Let $(\overline{g},K)$ be a perturbation of the Schwarzschild initial data set $(^S\overline{g},{^SK})$ on $\Sigma_0$, verifying the constraints (\ref{const}), and let $\{e_i\}^3_1$ be an orthonormal frame of $(\Sigma_0,\overline{g})$. We fix a reference frame $\{\partial_i\}^3_0$ in $\mathcal{M}^{1+3}=\{\Sigma_\tau\}_{\tau\in[0,T]}$, namely, the Schwarzschild orthonormal frame (\ref{Sframe}); Figure \ref{S_fol}. Let $\{e_i\}^3_0$, $e_0=\partial_0$, be a frame extension in $\mathcal{M}^{1+3}$ expressed in terms of $\partial_d$ via 
\begin{align}\label{Omain}
e_c=O^d_c\partial_d.
\end{align}
Consider now the (unique) metric $g$ for which $e_i$ is orthonormal, $g_{ab}:=m_{ab}=\text{diag}(-1,1,1,1)$, and the corresponding connection coefficients $(A_\nu)_{ij}=g(\nabla_{e_\mu}e_i,e_j)$. Then Proposition \ref{IDprop} asserts that the EVE (\ref{EVE}) for $g$, under the Lorentz gauge condition\footnote{We choose now a specific type based on the one satisfied by the Schwarzschild reference frame (\ref{divA2}).}
\begin{align}\label{gaugemain}
(\text{div}A_X)=(A_3)_{23}(A_2)_{ij},
\end{align}
reduce to the system of scalar equations
\begin{align}\label{system1}
\notag\boxdot(A_\nu)_{ij}=&\;{(A^{[\mu})_{\nu]}}^ke_k(A_\mu)_{ij}
+e^\mu([A_\mu,A_\nu])_{ij}
+e^\mu\big({(A_{[\mu})_{\nu]}}^k(A_k)_{ij}\big)\\
&+{(A^\mu)^k}_{[\mu}(F_{\nu] k})_{ij}
+([A^\mu,F_{\mu\nu}])_{ij}+e_\nu\big((A_3)_{23}(A_2)_{ij}\big)+e_\nu\big({(A^\mu)_\mu}^k(A_k)_{ij}\big)\\
\notag \partial_0(O^d_c)=&-O^b_c{^S{(A_{[0})_{b]}}^d}+{(A_{[0})_{c]}}^kO^d_k,\qquad \nu,i,j,c,d\in\{0,1,2,3\}
\end{align}
where $\boxdot:=-e^2_0+e^2_1+e^2_2+e^2_3$ and ${^S{(A_{[0})_{b]}}^d}=[\partial_0,\partial_b]^d$.

However, the system (\ref{system1}) has singular initial data in the Schwarzschild background which do not permit an energy approach directly. For this reason we recast the equations in a way that captures the closeness to the Schwarzschild spacetime.  
Let
\begin{align}\label{u}
(u_\nu)_{ij}:=(A_\nu)_{ij}-{^S(A_\nu)_{ij}}:\{\Sigma_\tau\}_{\tau\in[0,T]}\to\mathbb{R}
&&\nu,i,j\in\{0,1,2,3\},
\end{align}
where the components $^S(A_\nu)_{ij}$ are the Schwarzschild connection coefficients corresponding to the frame $\{\partial_i\}^3_0$ (\ref{Sframe}) and they are given by (\ref{A_mu_ij}). 
We are going to use these new functions to control the evolution of the
perturbed spacetime. 

Consider now the analogous system to (\ref{system1}) satisfied by the Schwarzschild components $^S(A_\nu)_{ij},\partial_c$. In view of the asymptotics (\ref{overasym1}), we define $\Gamma_q$ to be a smooth function satisfying the bound
\begin{align}\label{Gamma_q}
|\Gamma_q|\leq\frac{C_q}{r^q}&&|\partial^{(k)}\Gamma_q|\leq\frac{C_{q,k}}{r^{q+\frac{3}{2}k}},
\end{align}
for constants $C_q,C_{q,k}$ depending on $M>0$. Taking the difference of the two analogous systems we obtain a new system for the functions $(u_\nu)_{ij}, O^d_c-{I_c}^d$ written schematically in the form:
\begin{align}\label{boxu}
\notag h^{ab}\partial_a\partial_b(u_\nu)_{ij}
=&\;
O\Gamma_{\frac{3}{2}}\partial u+O\Gamma_3 u
+O\Gamma_{\frac{9}{2}}(O-I)
+O\Gamma_3\partial (O-I)\\
&+\Gamma_3u^2
+Ou\partial u
+u^3
+O\partial (O-I)\partial u\\
\partial_0(O^d_c-{I_c}^d)=&\;\Gamma_{\frac{3}{2}}(O-I)+(O-I)u
\notag+u,
\end{align}
where 
\begin{align}\label{h}
h^{ab}:=m^{cd}O^a_cO^b_d=g^{ab}
\end{align}
and each term in the RHS denotes some algebraic combination of finite number of terms of the depicted type (varying in $\nu,i,j$) where the particular indices do not matter.
\begin{remark}\label{stabred}
Evidently, the systems (\ref{system1}) and (\ref{boxu}) are equivalent. The benefit is that the  assumption on the perturbed spacetime, being close to Schwarzschild, implies that the functions $(u_\nu)_{ij}, O^d_c-{I_c}^d$ are now small and regular. 
Thus, we have reduced the evolutionary problem to solving the PDE-ODE system of
equations (\ref{boxu}).
However, the issue of singular initial data in (\ref{system1}) has become an issue of singularities in the coefficients of the resulting equations (\ref{boxu}), at $\tau=x=0$, which do not make it possible to apply the energy procedure in standard spaces; see also (\ref{notL2}). These singularities, in large part, are due to the intrinsic curvature blow up and cannot be gauged away; in particular the coefficients $\Gamma_3$ of the potential terms in (\ref{boxu}) correspond to the Schwarzschild curvature (\ref{K}). Some of the functions $\Gamma_q$ that appear in (\ref{boxu}), expressed in terms of Schwarzschild connection coefficients (\ref{A_mu_ij}) and their derivatives, are less singular than (\ref{Gamma_q}), but representatives of the exact bound do appear in all the terms. 
\end{remark}
\begin{remark}\label{Linfty1/r2}
Another crucial asymptotic behaviour that our method heavily depends on is that of the radius function $r$. According to (\ref{rasym}), we observe that the best
$L^\infty_{\Sigma_\tau}$ bound one could hope for the ratio $1/r^2$ is of the form
\begin{align}\label{1/r2}
\|\frac{1}{r^2}\|_{L^\infty(\Sigma_\tau)}\leq\frac{C}{\tau},
\end{align}
which obviously fails to be integrable in time $\tau\in[0,T]$, for any $T>0$. This fact lies at the heart of the difficulty of closing a Gronwall type estimate.
\end{remark}

\subsection{The weighted $H^s$ spaces}\label{wHs}

In order to study the well-posedness of (\ref{boxu})
we introduce certain weighted norms.
It turns out that the weights which yield
the desired flexibility in obtaining energy estimates are the following.
\begin{definition}\label{wsp}
Given $\alpha>0$ and $\tau\in[0,T]$, we define the (time dependent) weighted
Sobolev space $H^{s,\alpha}[\tau]$, as a subspace of the standard $H^s$ space on $\Sigma_\tau$ with the Schwarzschild induced volume form satisfying:
\begin{align}\label{Hsa}
H^{s,\alpha}[\tau]:\;\;\; u\in H^s(\Sigma_\tau),\;\;\;\|u\|_{H^{s,\alpha}[\tau]}^2:=\sum_{k\leq s}
\int_{\Sigma_\tau}\frac{[\partial^{(k)}u]^2}{r^{2\alpha-3(k-1)}}d\mu_{^S\overline{g}}<+\infty,
\end{align}
where by $\partial^{(k)}$ we denote any order $k$ combination of directional derivatives with respect
to the components $\partial_1,\partial_2,\partial_3$ of the
Schwarzschild frame (\ref{Sframe}). For convenience, we drop $\tau$ from the notation whenever the context is clear. 
\end{definition}
\begin{remark}
Observe that the weights in the norm $\|\cdot\|_{H^{s,\alpha}}$ in (\ref{Hsa}) blow up only at $\tau=0,x=0$.
For $\tau>0$ fixed, the weights are uniformly
bounded above by some positive constant $C_\tau$, which becomes infinite as $\tau\rightarrow0^+$. The dependence of the power $2\alpha-3(k-1)$ on the number $k$ of derivatives 
corresponds to the singularities in the coefficients
of the equation (\ref{boxu}).
\end{remark}
\begin{lemma}\label{Sobprop}
The weighted $H^{s,\alpha}$ spaces satisfy the properties:
\begin{align}\label{Sobprop1}
\notag H^{s_1,\alpha}\subset H^{s_2,\alpha}&&s_1<s_2\\
r^{-\frac{3}{2}l}u\in H^{s,\alpha-\frac{3}{2}l},&&\text{whenever $u\in H^{s,\alpha}$}\\
\notag \partial^{(k)}u\in H^{s-k,\alpha-\frac{3}{2}k}&&k\leq s,\; u\in H^{s,\alpha}
\end{align}
\end{lemma}
\begin{proof}
They are immediate consequences of Definition \ref{wsp} and and the fact that
\begin{align*}
|\partial_1(r^{-\frac{3}{2}l})|\leq Clr^{-\frac{3}{2}l-\frac{1}{2}}&&|\partial_2(r^{-\frac{3}{2}l})|=|\partial_3(r^{-\frac{3}{2}l})|=0,
\end{align*}
cf. (\ref{dr/dtau}), (\ref{Sframe}).
\end{proof}
\subsection{Local existence theorems}\label{main}

Let
\begin{align}\label{wen}
\notag\mathcal{E}(u,O;\alpha,T):=&
\sum_{\nu,i,j=0}^3\bigg[\sup_{\tau\in[0,T]}\big(\|(u_\nu)_{ij}\|^2_{H^{3,\alpha}}
+\|\partial_0(u_\nu)_{ij}\|^2_{H^{2,\alpha-\frac{3}{2}}}\big)\\
&+\int^T_0\big(\|(u_\nu)_{ij}\|^2_{H^{3,\alpha+1}}+\|\partial_0(u_\nu)_{ij}\|^2_{H^{2,\alpha-\frac{1}{2}}}
\big)d\tau\bigg]\\
&+\notag\sum_{c,d=0}^3\bigg[\sup_{\tau\in[0,T]}\|O^d_c-{I_c}^d\|^2_{H^{3,\alpha+\frac{3}{2}}}
+\int^T_0\|O^d_c-{I_c}^d\|^2_{H^{3,\alpha+\frac{5}{2}}}d\tau
\bigg]
\end{align}
be the total weighted energy of the functions $(u_\nu)_{ij},O^d_c-{I_c}^d$ defined in $\{\Sigma_\tau\}_{\tau\in[0,T]}$ (\ref{fol}), Figure \ref{fol}, the backward domain of dependence  of $\Sigma_0$ with respect to the metric $g$ we are solving for. Since the actual domain depends on the unknown solution, it will be fully determined in the end; see Section \ref{fixedpt}. For brevity we denote by
\begin{align}\label{initwen}
\mathcal{E}_0:=&
\sum_{\nu,i,j\in\{0,1,2,3\}}\bigg[\|(u_\nu)_{ij}(\tau=0)\|^2_{H^{3,\alpha}}+
\|\partial_0(u_\nu)_{ij}(\tau=0)\|^2_{H^{2,\alpha-\frac{3}{2}}}\bigg]\\
\notag&+\sum_{c,d\in\{0,1,2,3\}}\|O^d_c-{I_c}^d\|^2_{H^{3,\alpha+\frac{3}{2}}(\Sigma_0)}
\end{align}
the energy at the initial singular slice $\Sigma_0$. 
\begin{remark}\label{ufanal}
It would seem more natural if the components $(u_\nu)_{ij}$ were lying in $H^{2,\alpha}$, that is, one derivative less than $O^d_c-{I_c}^d$. However, this would create additional technical difficulties, which we choose to avoid by treating both $(u_\nu)_{ij}$, $O^d_c-{I_c}^d$ at the same footing (number of derivatives), since it is an issue of the structure of the system (\ref{boxu}) and not a singularity issue.
\end{remark}
The following theorem is our first main local well-posedness result for the system (\ref{boxu}), whose proof occupies Section \S\ref{fixedpt}.
\begin{theorem}\label{mainthm}
There exist $\alpha>0$ sufficiently large and $\varepsilon>0$ small such that if
\begin{align}\label{initf}
\mathcal{E}_0<+\infty&&\|O^d_c-{I_c}^d\|_{L^\infty(\Sigma_0)}<\varepsilon,&&\;c,d=0,1,2,3,
\end{align}
then the system
(\ref{boxu}) admits a unique solution, up to some small time $T=T(\mathcal{E}_0,\alpha)>0$, in the spaces
\begin{align}\label{solsp}
\notag(u_\nu)_{ij}\in&\;C([0,T];H^{3,\alpha})\cap L^2([0,T];H^{3,\alpha+1})\qquad\qquad\nu,i,j\in\{0,1,2,3\}\\
\partial_0(u_\nu)_{ij}\in&\;C([0,T];H^{2,\alpha-\frac{3}{2}})\cap L^2([0,T];H^{2,\alpha-\frac{1}{2}})\\
\notag O^d_c-{I_c}^d\in&\;C([0,T];H^{3,\alpha+\frac{3}{2}})\cap
L^2([0,T];H^{3,\alpha+\frac{5}{2}})\qquad\qquad c,d\in\{0,1,2,3\}
\end{align}
\end{theorem}
\begin{remark}\label{h-m}
$(i)$ The second part of condition (\ref{initf}), $\varepsilon>0$ small, is necessary for the equation (\ref{boxu}) to be hyperbolic,
yielding sufficient pointwise control on the $h^{ab}$'s (\ref{h})
\begin{align}\label{hyp}
|h^{bb}-m^{bb}|<\frac{1}{2}&&|h^{bc}|\leq C\varepsilon^2,\qquad\text{$b,c=0,1,2,3$, $b\neq c$}.
\end{align}
It could be obviously replaced by the stronger assumption that $\mathcal{E}_0<\varepsilon$, since the energy $\mathcal{E}(u,O;\alpha,T)$ controls the $L^\infty$ norm of $u,O$ by standard Sobolev embedding.\\
$(ii)$ How large the exponent $\alpha$ has to be depends on the coefficients of the system (\ref{boxu}).
In the final inequalities in \S\ref{fixedpt} $\alpha>0$ is picked large enough so that certain `critical' terms can be absorbed in the LHS and the estimates can close.
\end{remark}
The above theorem is a local well-posedness result for the system (\ref{boxu}). Imposing now the proper conditions on the initial data set of (\ref{boxu}), the solution (\ref{solsp}) yields
a solution of (\ref{system1}) which in turn corresponds to an Einstein vacuum spacetime (\ref{EVE}).
\begin{theorem}\label{mainthm2}
Let $\alpha,\varepsilon$ be such as in Theorem \S\ref{mainthm}
and let $(\Sigma_0,\overline{g},K)$
be an initial data set for the Einstein vacuum equations (\ref{EVE}) satisfying the constraints (\ref{const}),
such that the components
\begin{align}\label{initcond2}
(u_\nu)_{ij}\in H^{3,\alpha}(\Sigma_0)&&\nu,i,j=1,2,3,
\end{align}
\begin{align}\label{initf2}
O^d_c-{I_c}^d\in H^{3,\alpha+\frac{3}{2}}(\Sigma_0)
&&\|O^d_c-{I_c}^d\|_{L^\infty(\Sigma_0)}<\varepsilon&&c,d=1,2,3,
\end{align}
computed with respect to an orthonormal frame $\{e_i\}^3_1$
on $(\Sigma_0,\overline{g})$, and
\begin{align}\label{diffk}
(u_i)_{0j}(\tau=0):=K_{ij}-{^SK_{ij}}\in H^{3,\alpha}(\Sigma_0)&&i,j=1,2,3.
\end{align}
Then, there exists a solution $g$ to the EVE (\ref{EVE}) in the backward region to $\Sigma_0$, foliated by $\{\Sigma_\tau\}_{t\in[0,T]}$,
with induced initial data set $(\overline{g},K)$ on $\Sigma_0$
and an orthonormal frame extension $\{e_i\}^3_0$ for which
the corresponding (spacetime) functions $(u_\nu)_{ij},O^d_c-I^d_c$ (\ref{u}),(\ref{Omain}) lie in the spaces (\ref{solsp}).\\
If in addition $O^d_c-{I_c}^d\in C([0,T];H^{4,\alpha+\frac{3}{2}})$, $c,d=1,2,3$, then
the Einsteinian vacuum development is unique up to isometry.
\end{theorem}
The fact that such (non-spherically symmetric) initial data sets $(\Sigma_0,\overline{g},K)$ exist,
in compliance with Theorem \ref{mainthm2}, is shown in \S\ref{singconst}.
\begin{proof}[Proof of Theorem \ref{mainthm2}]
We want to invoke Theorem \ref{mainthm}. For this purpose, we prescribe
initial data for the system (\ref{boxu}):\\
$(i)$ The components (\ref{initcond2}), (\ref{initf2}), (\ref{diffk})
are given.\\
$(ii)$ Since in the beginning of \S\ref{pertsp} we assumed $e_0=\partial_0$ and since $\{e_i\}^3_1$ is initially tangent to
$\Sigma_0$, we set
\begin{align}\label{restinitf}
O^b_0(\tau=0)={I_0}^b&&O^0_a(\tau=0)={I_a}^0&&a,b=0,1,2,3.
\end{align}
$(iii)$ We (freely) assign\footnote{The functions $(u_0)_{ab}(\tau=0)$ or equivalently $(A_0)_{ab}(\tau=0)$ fix the $\partial_0$ derivative of the frame $\{e_i\}_0^3$ on $\Sigma_0$; see Lemma \ref{gaugemap} and Remark \ref{initfixed}.}
\begin{align}\label{restdir}
(u_0)_{ab}(\tau=0):=(A_0)_{ab}-{^S(A_0)_{ab}}\in H^{3,\alpha}(\Sigma_0),\qquad\qquad a,b=0,1,2,3.
\end{align}
Once we have prescribed the above, the components $\partial_0(u_\nu)_{ij}(\tau=0)$ are fixed by the assumption (\ref{initA}) on the initial data of the original system (\ref{system1}); see Remark \ref{initfixed}. Indeed, subtracting the corresponding Schwarzschild components from (\ref{initfixed1}),(\ref{initfixed2}), which obvisouly satisfy the same initial relations, cf. (\ref{divA2}), we obtain schematically:
\begin{align}\label{e_0(u)init}
\partial_0(u_\nu)_{ij}=O\partial_au+\Gamma_\frac{3}{2}u
+\Gamma_3(O-I)+u^2&&\text{on $\Sigma_0$,\; $a=1,2,3.$}
\end{align}
By (\ref{Sobprop1}) and standard Sobolev embedding we conclude that
\begin{align}\label{restneum}
\partial_0(u_\nu)_{ij}(\tau=0)\in H^{2,\alpha-\frac{3}{2}}&&\nu,i,j=0,1,2,3.
\end{align}
Thus, the assumption (\ref{initf}) is verified and Theorem \ref{mainthm} can be invoked. From Proposition \ref{IDprop} it follows that the solution (\ref{solsp}) of (\ref{boxu}) and hence of (\ref{system1}) yields indeed an Einstein vacuum spacetime $(\{\Sigma_\tau\}_{\in[0,T]},g)$.

To prove uniqueness (up to isometry) we rely on the uniqueness statement in Theorem \ref{mainthm}. Suppose there is another Einsteinian vacuum development $(\tilde{\mathcal{M}}^{1+3},\tilde{g})$ of the initial data set $(\Sigma_0,\overline{g},K)$, diffeomorphic to $\{\Sigma_\tau\}_{\tau\in[0,T]}$, satisfying the hypothesis (\ref{initcond2}), (\ref{initf2}), (\ref{diffk}); defined by pulling back the relevant quantities through the preceding diffeomorphism, taking differences etc. 
In order to use the uniqueness statement in Theorem \ref{mainthm}, we need the two spacetimes to have the same initial data for the system (\ref{boxu}).
The part of the initial data set given by the assumptions in the statement of Theorem \ref{mainthm2} is of course identical for both spacetimes.
The remaining components that we want to agree, other than the $(\tilde{u}_0)_{ab}(\tau=0)$'s, as noted in the previous paragraph, can be fixed by condition (\ref{initA}). Therefore, we get identical initial data components for the system (\ref{boxu}) by constructing 
a Lorentz gauge frame (\ref{gaugemain}) $\{\tilde{e}_i\}^3_0$ for $\tilde{g}$, which is initially equal to $\{e_i\}^3_0$ on $\Sigma_0$ and such that $(\tilde{u}_0)_{ab}(\tau=0)=(u_0)_{ab}(\tau=0)$ as well; see Lemma \ref{gaugemap}. The only assumption to be verified is the well-posedness of the system (\ref{boxO2}) for functions in the solution spaces (\ref{solsp}), after taking differences with the equation for the frame $\{e_i\}^3_0$. However, this falls in the category of the system (\ref{boxu}) [in fact simpler, being semilinear] to which Theorem \ref{mainthm} can be applied. The extra derivative that we have to assume in order to close, $\tilde{O}^d_c-{I_c}^d\in H^{4,\alpha+\frac{3}{2}}$, is due to the $\text{div}A$ term in the RHS of (\ref{boxO2}). 
\end{proof}

\section{Proof of Theorem \ref{mainthm}}\label{fixedpt}

Throughout this section we will use the notation $X\lesssim Y$ to denote an inequality between the
quantities $X,Y$ of the form
$X\leq CY$, where $C$ is an absolute positive constant depending only on the Schwarzschild mass $M>0$. The same
for the standard notation $O(X)$, for a quantity bounded by $|O(X)|\leq CX$, $X>0$. Furthermore, all the
estimates regard only the Schwarzschild region foliated by $\{\Sigma_\tau\}_{\tau\in[0,T]}$; Figure \ref{genS_fol}.

\subsection{Proof outline}\label{outlpf}

We prove Theorem \ref{mainthm} via a contraction mapping argument. First we establish an energy estimate in the relevant weighted $H^3$ spaces in \S\ref{Enest}. Then we obtain a contraction, in \S\ref{Contr}, in the corresponding spaces of one derivative less, see (\ref{contr}), which together with the energy estimate yield the desired solution (\ref{solsp}).\\
To derive these estimates we have to eliminate some {\it critical} terms which are generated due to the singularities in the coefficients of the equations, having larger weights than the ones in the norm (\ref{Hsa}), and which prevent us from closing (see Propositions \ref{mainenest},\ref{maincontrest}). This is where the role of the weights (\ref{Hsa}) comes in. The parameter $\alpha>0$ helps generate critical terms with a favourable sign. Being large enough, but finite, $\alpha$ provides an overall negative sign for the critical terms, hence, rendering them removable from the RHS of the final inequalities. This enables us to close the estimates and complete the proof. The precise asymptotics of the singularities in the coefficients of the equations (\ref{boxu}), at $\tau=x=0$, and the opening up rate of the radius function $r$ in $\tau>0$ play a crucial role here.\footnote{If we were to tweak the leading orders just by $\epsilon>0$, the previous procedure would fail no matter how large $\alpha>0$ is to begin with.}

\subsection{Basic estimates}

Let $v$
be a scalar function defined on $\Sigma_\tau$,
represented by
\begin{align}\label{scal}
v\circ\psi_\tau:U_\tau\to\mathbb{R},
\end{align}
where $\psi_\tau:U_\tau\to\Sigma_\tau$ is the $(x,\theta,\phi)$ coordinate chart.
We recall some standard inequalities: the classical Sobolev embedding of $H^2(U)$
in $L^\infty(U)$
and the interpolation inequality
\begin{align}\label{G-N}
\|v\|_{L^4(U)}\leq C\|v\|^\frac{1}{4}_{L^2(U)}\|\nabla v\|^\frac{3}{4}_{L^2(U)},
\end{align}
for a bounded domain $U\subset\mathbb{R}^3$ with (piecewise) $C^2$ boundary. In the following proposition $v$ is assumed to be regular enough such that the RHSs make sense. 
\begin{proposition}\label{adaptineq}
For a general function $v:\Sigma_\tau\to\mathbb{R}$, $\tau\in[0,T]$,
with the appropriate regularity, the following inequalities hold:\\ 
The $L^\infty$ bound
\begin{align}\label{Linfty}
\|\frac{v}{r^k}\|_{L^\infty(\Sigma_\tau)}\lesssim(k+1)^2\|v\|_{H^{2,k+3+\frac{1}{4}}(\Sigma_\tau)}
\end{align}
and the $L^4$ estimate
\begin{align}\label{L4}
\|\frac{v}{r^k}\|_{L^4(\Sigma_\tau)}
\lesssim (k+1)^\frac{3}{4}\|v\|_{H^{1,k+1+\frac{1}{4}}(\Sigma_\tau)}.
\end{align}
\end{proposition}
\begin{proof}
From the embedding $H^2(U_\tau)\hookrightarrow L^\infty(U_\tau)$ we have 
\begin{align*}
\|\frac{v}{r^k}\|_{L^\infty(\Sigma_\tau)}\overset{(\ref{scal})}{=}\|\frac{v}{r^k}\circ\psi_\tau\|_{L^\infty(U_\tau)}\lesssim&\;
\|\frac{v}{r^k}\circ\psi_\tau\|_{H^2(U_\tau)}\\
\tag{substituting (\ref{dmu}) and the frame (\ref{Sframe})}
\lesssim&\;(k+1)^2\|v\|_{H^{2,k+3+\frac{1}{4}}(\Sigma_\tau)}.
\end{align*}
We argue similarly in the case of (\ref{L4}).
\end{proof}
\subsection{Energy estimate in $H^{3,\alpha}$}\label{Enest}

We define now the mapping, which really corresponds to an iterative process. 
Let $\{\overline{u},\overline{O}\}:=\{(\overline{u}_\nu)_{ij},\overline{O}^d_c:\nu,i,j,c,d=0,1,2,3\}$ be a set of functions in the solution spaces (\ref{solsp}), verifying $|\overline{O}^d_c-{I_c}^d|<\varepsilon$ initially on $\Sigma_0$. Let $T>0$ be sufficiently small such that
\begin{align}\label{venest}
\mathcal{E}(\overline{u},\overline{O};\alpha,T)\leq2\mathcal{E}_0.
\end{align}
We also assume\footnote{Any assumptions that we make on the functions 
$ \overline{u},\overline{O}$, we must derive for the next set of functions $u,f$ below.}
\begin{align}\label{norme_0barf}
\|\partial_0(\overline{O}^d_c)\|^2_{H^{2,\alpha}[\tau]}\lesssim\mathcal{E}^2_0+\mathcal{E}_0&&\forall\tau\in[0,T],\;\;c,d=0,1,2,3.
\end{align}

{\it Iteration step}: Consider the following linear version of the system
(\ref{boxu}), where we replace the functions $u,O$ in the following specific terms by the corresponding ones from the set $\{\overline{u},\overline{O}\}$:
\begin{align}\label{modboxu}
\notag \overline{h}^{ab}\partial_a\partial_b(u_\nu)_{ij}
=&\;
\overline{O}\Gamma_{\frac{3}{2}}\partial u+\overline{O}\Gamma_3 u
+\overline{O}\Gamma_{\frac{9}{2}}(O-I)
+\overline{O}\Gamma_3\partial (O-I)\\
&+\Gamma_3\overline{u}^2
+\overline{O}\overline{u}\partial \overline{u}
+\overline{u}^3
+\overline{O}\partial (\overline{O}-I)\partial \overline{u}\\
\partial_0(O^d_c-{I_c}^d)=&\;\Gamma_{\frac{3}{2}}(O-I)+(\overline{O}-I)\overline{u}
\notag+u,
\end{align}
where $\overline{h}^{ab}=m^{cd}\overline{O}^a_c\overline{O}^b_d$. Observe that we kept in the RHS of (\ref{modboxu}) the functions $u,O$ attached to the most singular coefficients of the system. This is actually very important to our strategy. Had we replaced them with the corresponding functions $\overline{u},\overline{O}$ as well, it would not be plausible to derive a weighted energy estimate.

We assume now there exists a solution $(u_\nu)_{ij},O^d_c-{I_c}^d$ of (\ref{modboxu}) lying in the solution space (\ref{solsp}). The existence of such a solution is based mainly on the energy estimate we will derive below and a standard duality argument which we omit. \\
{\it Claim}: For a chosen large enough $\alpha>0$ and $T>0$ sufficiently small (depending on $\mathcal{E}_0,\alpha$) the following estimate holds
\begin{align}\label{uenest}
\mathcal{E}(u,O;\alpha,T)\leq2\mathcal{E}_0.
\end{align}
The preceding $H^3$-weighted energy estimate, cf. (\ref{wen}), will be used in the next subsection to close the contraction argument that yields the existence and uniqueness of the solution (\ref{solsp}) to (\ref{boxu}). Now we begin the proof of (\ref{uenest}):

First note that by the fundamental theorem of calculus, following a $\partial_0$ integral curve and employing (\ref{Linfty}), we readily obtain from our initial assumptions and (\ref{norme_0barf}) the pointwise bound
\begin{align}\label{Linftybarf}
\sup_{\tau\in [0,T]}\|\overline{O}- I|\|_{L^\infty(\Sigma_\tau)}\leq \varepsilon+CT\mathcal{E}_0<2\varepsilon,
\end{align}
provided $\alpha\geq\frac{1}{2}+3+\frac{1}{4}$ and $T<\frac{\varepsilon}{C\mathcal{E}_0}$. 

All the more, directly from the ODE in (\ref{modboxu}) we deduce the estimate: [applying the bounds (\ref{Linfty}), (\ref{venest}) to $(\overline{O}- I)\overline{u}$ and employing the asymptotics (\ref{Gamma_q})] 
\begin{align}\label{norme_0f}
\|\partial_0(O^d_c)\|^2_{H^{2,\alpha}[\tau]}\lesssim \mathcal{E}_0^2+\|O-I\|^2_{H^{2,\alpha+\frac{3}{2}}[\tau]}+\|u\|^2_{H^{2,\alpha}[\tau]},
\end{align}
for all $\tau\in[0,T]$, $c,d=0,1,2,3$.\footnote{This estimate, together with (\ref{uenest}) in the end, imply the analogue of (\ref{norme_0barf}) for the functions $\partial_0(O^d_c)$.}

We derive (\ref{uenest}) in the backward domain of dependence of $\Sigma_0$ w.r.t. the metric $(g_{ab})_{\overline{u}}:=g_{\overline{u}}(\partial_a,\partial_b)$, $a,b=0,1,2,3$, whose inverse is given by $g^{ab}_{\overline{u}}:=\overline{h}^{ab}$; compare to (\ref{h}). The boundary of the domain is the backward incoming
$g_{\overline{u}}$-null hypersurface $\mathcal{N}^{\overline{u}}$ emanating from $\partial\Sigma_0$ (Figure \ref{fol_v}). We foliate the domain by the $\tau=\text{const.}$ hypersurfaces $\Sigma^{\overline{u}}_\tau$ inside $\mathcal{N}^{\overline{u}}$. Let $\rho$ be the scalar function defined near $\mathcal{N}^{\overline{u}}$ via
\begin{align}\label{rhoconst}
\rho(\mathcal{C}^{\overline{u}}_\tau):=T-\tau,
\end{align}  
where $\mathcal{C}^{\overline{u}}_\tau$ is the cylinder obtained from the flow of $\partial\Sigma_\tau^{\overline{u}}$ backwards along the integral curves of $\partial_0$. Using $\rho$ we may write each leaf of the foliation as 
\begin{align}\label{folSigma}
\Sigma_\tau^{\overline{u}}=\bigcup_{t^*\in[\tau,T]}\{\rho_\tau=T-t^*\}\bigcup B_\tau&&
\tau\in[0,T],
\end{align}
where $\rho_\tau:=\rho\big|_{\Sigma_\tau^{\overline{u}}}$ and $B_\tau$ 
is simply the projection of $\Sigma_T^{\overline{u}}$ onto $\Sigma_\tau^{\overline{u}}$ through 
the integral curves of $\partial_0$. 
\begin{figure}[h!]
  \centering
    \includegraphics[width=0.7\textwidth]{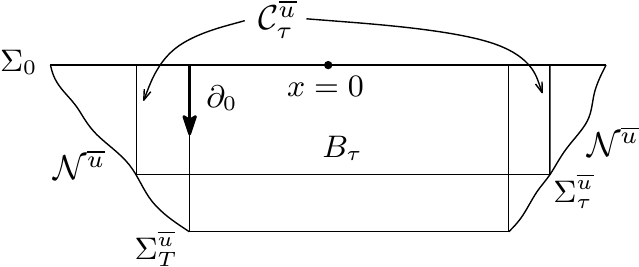}
      \caption{}
\label{fol_v}
\end{figure}
Since by definition $\rho+\tau-T$ is zero on $\mathcal{N}^{\overline{u}}$, 
it follows that the $g_{\bar{u}}$-gradient of $\rho+\tau-T$, on $\mathcal{N}^{\overline{u}}$, lies on the hypersurface itself and furthermore it is $g_{\bar{u}}$-null, i.e., $\rho$ satisfies the eikonal equation
\begin{align}\label{rhoeik}
\big|\nabla_{g_{\overline{u}}}(\rho+\tau-T)\big|^2_{g_{\overline{u}}}=&\;\overline{h}^{AB}\partial_A(\rho)\partial_B(\rho)
+\Omega^{-2}\overline{h}^{00}+2\Omega^{-1}\overline{h}^{A0}\partial_A(\rho)\\
\notag=&\;0\qquad\qquad\text{on $\mathcal{N}^{\overline{u}}$},
\end{align}
where $A,B=1,2,3$.

In this regard, we define the following adapted energy, which controls the part of the total energy (\ref{wen}) that refers to $u$:
\begin{align}\label{E_s,alpha}
E_{s+1,\alpha}[u](\tau):=&\frac{1}{2}\sum_{\nu,i,j}\sum_{|J|\leq s}\int_{\Sigma_\tau^{\overline{u}}}\bigg[-\overline{h}^{00}\frac{\big[\partial_0(u_\nu)_{ij,J}\big]^2}{r^{2\alpha-3|J|}}\\
&+\overline{h}^{AB}\frac{\partial_A(u_\nu)_{ij,J}}{r^{\alpha-\frac{3}{2}|J|}}
\frac{\partial_B(u_\nu)_{ij,J}}{r^{\alpha-\frac{3}{2}|J|}}
\notag+\frac{(u_\nu)_{ij,J}^2}{r^{2\alpha-3(|J|-1)}}\bigg]d\mu_{^S\overline{g}},
\end{align}
where $(u_\nu)_{ij,J}:=\partial^{(J)}(u_\nu)_{ij}$ and $J$ is a spatial multi-index (containing only directions $\partial_1,\partial_2,\partial_3$).
It is evident from (\ref{Linftybarf}), $\overline{h}^{ab}=m^{cd}\overline{O}^a_c\overline{O}^b_d$, that 
$E_{3,\alpha}$ is equivalent to the weighted 
$H^{3,\alpha}\times H^{2,\alpha-\frac{3}{2}}$ norm of $u$ on $\Sigma_\tau^{\overline{u}}$. 

We summarize in the following proposition the main energy estimates derived below.
\begin{proposition}\label{mainenest}
The following two energy estimates hold:
\begin{align}\label{E_sest}
\notag&\partial_\tau E_{3,\alpha}[{ u}]+8M^2e^{-1}(1-\tau)\alpha E_{3,\alpha+1}[{ u}]\\
\lesssim&\; (\mathcal{E}_0^\frac{1}{2}+\mathcal{E}_0+\alpha^2+\alpha^3\mathcal{E}_0)E_{3,\alpha}[{ u}]+E_{3,\alpha+1}[{ u}]+\mathcal{E}_0\|{ O-I}\|^2_{H^{3,\alpha+\frac{3}{2}}}\\
&\notag
+\|{ O-I}\|^2_{H^{3,\alpha+\frac{5}{2}}}
+\alpha^3\mathcal{E}_0^2+\mathcal{E}_0^3
\end{align}
\begin{align}\label{Oest}
&\frac{1}{2}\partial_\tau\sum_{c,d}\|O^d_c-{I_c}^d\|^2_{H^{3,\alpha+\frac{3}{2}}}
+4M^2e^{-1}(1-\tau)\alpha\sum_{c,d}\|O^d_c-{I_c}^d\|^2_{H^{3,\alpha+\frac{5}{2}}}\\
\notag\lesssim&\;\|{ O-I}\|^2_{H^{3,\alpha+\frac{5}{2}}}
+E_{3,\alpha+1}[{ u}]+\mathcal{E}_0^2,
\end{align}
for all $\tau\in(0,T)$.
\end{proposition}
The overall energy estimate (\ref{uenest}) follows from Proposition \ref{mainenest}:
Adding (\ref{E_sest}), (\ref{Oest}) we wish to close the estimate by employing the standard Gronwall lemma.
However, this is not possible in general, because of 
the critical energies in the RHS, having larger weights than the ones differentiated in the LHS, namely, $E_{3,\alpha+1}[{ u}]$, $\|{ O-I}\|^2_{H^{3,\alpha+\frac{5}{2}}}$ instead of $E_{3,\alpha}[{ u}],\|{ O-I}\|^2_{H^{3,\alpha+\frac{3}{2}}}$. It is precisely at this point that the role of the weights we introduced is revealed.
Choosing $\alpha>0$ large enough to begin with, how large depending on the 
constants in the above inequalities, we absorb the critical terms 
\begin{align*}
E_{3,\alpha+1}[{ u}],\|{ O-I}\|^2_{H^{3,\alpha+\frac{5}{2}}}
\end{align*}
in the LHS and then the standard Gronwall lemma applies to give (\ref{uenest}).
\begin{proof}[Proof of (\ref{E_sest})]
Let
\begin{align}\label{topE}
P_{J,\alpha}:=\frac{1}{2}\bigg[-\overline{h}^{00}\frac{\big[\partial_0(u_\nu)_{ij,J}\big]^2}{r^{2\alpha-3|J|}}
+\overline{h}^{AB}\frac{\partial_A(u_\nu)_{ij,J}}{r^{\alpha-\frac{3}{2}|J|}}
\frac{\partial_B(u_\nu)_{ij,J}}{r^{\alpha-\frac{3}{2}|J|}}
+\frac{(u_\nu)_{ij,J}^2}{r^{2\alpha+3-3|J|}}\bigg],
\end{align}
for any spatial multi-index $J$ with $|J|\leq2$; recall $(u_\nu)_{ij,J}:=\partial^{(J)}(u_\nu)_{ij}$.
It follows from (\ref{folSigma}) and the coarea formula that
\begin{align}\label{E/dtau}
\partial_\tau\int_{\Sigma_\tau^{\overline{u}}} P_{J,\alpha}d\mu_{^S\overline{g}}=
&-\int_{\partial\Sigma_\tau^{\overline{u}}} \frac{P_{J,\alpha}}
{|{^S\overline{\nabla}}\rho|}dS
+\int_{\Sigma_\tau^{\overline{u}}} \partial_\tau P_{J,\alpha}d\mu_{^S\overline{g}}\\
\notag&+\int_{\Sigma_\tau^{\overline{u}}} P_{J,\alpha}\partial_\tau d\mu_{^S\overline{g}},
\end{align}
where $^S\overline{\nabla}\rho$ stands for the gradient of $\rho$ with respect to the intrinsic connection on $(\Sigma_\tau,{^S\overline{g}})$ and $dS$ is the Schwarzschild induced volume form on $\partial\Sigma_\tau^{\overline{u}}$. Note that the boundary term in (\ref{E/dtau}) has a favourable sign. Since $\mathcal{N}^{\overline{u}}$ is $g_{\overline{u}}$\,-incoming null, we show that the sum of all arising boundary terms has a good sign and therefore can be dropped in the end. To analyse the last two terms in (\ref{E/dtau}), we recall the $\partial_\tau$ differentiation
formulas of the radius function $r$ (\ref{dr/dtau}), the estimate on volume form 
$d\mu_{^S\overline{g}}$ (\ref{dmu/dtau}) and the commutator relation $[\partial_0,\partial_B]={^S(A_{[0})_{B]}}^c\partial_c\overset{(\ref{overasym1})}{=}\Gamma_{\frac{3}{2}}\partial$:
\begin{align}\label{intP/dtau}
\notag&\int_{\Sigma_\tau^{\overline{u}}} \partial_\tau P_{J,\alpha}d\mu_{^S\overline{g}}
+\int_{\Sigma_\tau^{\overline{u}}} P_{J,\alpha}\partial_\tau d\mu_{^S\overline{g}}\\
\notag=& -8M^2(1-\tau)\alpha\int_{\Sigma_\tau^{\overline{u}}} e^{-\frac{r}{2M}}P_{J,\alpha+1}d\mu_{^S\overline{g}}
+\int_{\Sigma_\tau^{\overline{u}}} P_{J,\alpha}O(\frac{1}{r^2})d\mu_{^S\overline{g}}\\
&+\frac{1}{2}\int_{\Sigma_\tau^{\overline{u}}}\Omega\bigg[
-\partial_0(\overline{h}^{00})\frac{\big[\partial_0(u_\nu)_{ij,J}\big]^2}{r^{2\alpha-3|J|}}
+\partial_0(\overline{h}^{AB})\frac{\partial_A(u_\nu)_{ij,J}}{r^{\alpha-\frac{3}{2}|J|}}
\frac{\partial_B(u_\nu)_{ij,J}}{r^{\alpha-\frac{3}{2}|J|}}\bigg]d\mu_{^S\overline{g}}\\
\notag&+\int_{\Sigma_\tau^{\overline{u}}}\Omega\bigg[
-\overline{h}^{00}\frac{\partial_0(u_\nu)_{ij,J}\partial_0^2(u_\nu)_{ij,J}}{r^{2\alpha-3|J|}}
+\overline{h}^{AB}\frac{\partial_A(u_\nu)_{ij,J}}{r^{\alpha-\frac{3}{2}|J|}}
\frac{\partial_B\partial_0(u_\nu)_{ij,J}}{r^{\alpha-\frac{3}{2}|J|}}
\bigg]d\mu_{^S\overline{g}}\\
\notag&
+\int_{\Sigma_\tau^{\overline{u}}}\Omega
\overline{h}^{AB}\frac{\partial_A(u_\nu)_{ij,J}}{r^{\alpha-\frac{3}{2}|J|}}
\frac{\Gamma_{\frac{3}{2}}\partial(u_\nu)_{ij,J}}{r^{\alpha-\frac{3}{2}|J|}}d\mu_{^S\overline{g}}
+\int_{\Sigma_\tau^{\overline{u}}} \Omega\frac{(u_\nu)_{ij,J}\partial_0(u_\nu)_{ij,J}}{r^{2\alpha+3-|J|}}d\mu_{^S\overline{g}}
\end{align}
The first term on the LHS of (\ref{intP/dtau}) is critical having a favourable sign of magnitude $\alpha$. We use this term alone to absorb all arising critical terms in the process. 
Recall $|\overline{h}|=|\overline{O}^2|\leq1$, cf. (\ref{Linftybarf}), and the asymptotics (\ref{Gamma_q}). Also, applying (\ref{Linfty}) to $\partial_0\overline{ h}$ and (\ref{norme_0barf}) we derive
\begin{align*}
|\Omega \partial_0({ h})|\lesssim \mathcal{E}({ \overline{u},\overline{O}};\alpha,T)^\frac{1}{2},&&
\Omega\lesssim\frac{1}{r^{\frac{1}{2}}}, &&|\Gamma_\frac{3}{2}|\lesssim\frac{1}{r^\frac{3}{2}}.
\end{align*}
Hence, by Cauchy's inequality and (\ref{venest}) we have
\begin{align}\label{Pest1}
\notag&\frac{1}{2}\int_{\Sigma_\tau^{\overline{u}}}\Omega\bigg[
-\partial_0(\overline{h}^{00})\frac{\big[\partial_0(u_\nu)_{ij,J}\big]^2}{r^{2\alpha-3|J|}}
+\partial_0(\overline{h}^{AB})\frac{\partial_A(u_\nu)_{ij,J}}{r^{\alpha-\frac{3}{2}|J|}}
\frac{\partial_B(u_\nu)_{ij,J}}{r^{\alpha-\frac{3}{2}|J|}}\bigg]d\mu_{^S\overline{g}}\\
\notag&
+\int_{\Sigma_\tau^{\overline{u}}}\Omega
\overline{h}^{AB}\frac{\partial_A(u_\nu)_{ij,J}}{r^{\alpha-\frac{3}{2}|J|}}
\frac{\Gamma_\frac{3}{2}\partial(u_\nu)_{ij,J}}{r^{\alpha-\frac{3}{2}|J|}}d\mu_{^S\overline{g}}
+\int_{\Sigma_\tau^{\overline{u}}} \Omega\frac{(u_\nu)_{ij,J}\partial_0(u_\nu)_{ij,J}}{r^{2\alpha+3-|J|}}d\mu_{^S\overline{g}}\\
\lesssim&\;\mathcal{E}({ \overline{u},\overline{O}};\alpha,T)^\frac{1}{2}
E_{3,\alpha}[{ u}]
+E_{3,\alpha+1}[{ u}]\\
\notag\lesssim&\;\mathcal{E}_0^\frac{1}{2}
E_{3,\alpha}[{ u}]+E_{3,\alpha+1}[{ u}]
\end{align}
For the next term we proceed by integrating by parts\footnote{We integrate by parts using the spatial part of the Schwarzschild frame $\partial_1,\partial_2,\partial_3$. Doing so we pick up connection coefficients, since it is not covariant IBP.} (IBP), 
denoting by $N:=\overline{g}^{BB}N_B\partial_B$ 
the outward unit normal on $\partial\Sigma_\tau^{\overline{u}}$ w.r.t. Schwarzschild metric $\overline{g}$ on $\Sigma_\tau^{\overline{u}}$:
\begin{align}\label{Pest2}
\notag&\int_{\Sigma_\tau^{\overline{u}}}\Omega\bigg[
-\overline{h}^{00}\frac{\partial_0(u_\nu)_{ij,J}\partial_0^2(u_\nu)_{ij,J}}{r^{2\alpha-3|J|}}
+\overline{h}^{AB}\frac{\partial_A(u_\nu)_{ij,J}}{r^{\alpha-\frac{3}{2}|J|}}
\frac{\partial_B\partial_0(u_\nu)_{ij,J}}{r^{\alpha-\frac{3}{2}|J|}}
\bigg]d\mu_{^S\overline{g}}\\
=&-\int_{\Sigma_\tau^{\overline{u}}}\Omega\frac{\partial_0(u_\nu)_{ij,J}}{r^{\alpha-\frac{3}{2}|J|}}\bigg[
\overline{h}^{00}\frac{\partial_0^2(u_\nu)_{ij,J}}{r^{\alpha-\frac{3}{2}|J|}}
+\overline{h}^{AB}\frac{\partial_B\partial_A(u_\nu)_{ij,J}}{r^{\alpha-\frac{3}{2}|J|}}
\bigg]d\mu_{^S\overline{g}}\\
\notag&+\int_{\partial\Sigma_\tau^{\overline{u}}}\Omega
\overline{h}^{AB}\frac{\partial_A(u_\nu)_{ij,J}}{r^{\alpha-\frac{3}{2}|J|}}
\frac{\partial_0(u_\nu)_{ij,J}}{r^{\alpha-\frac{3}{2}|J|}}
N_B dS\\
&\notag
-\int_{\Sigma_\tau^{\overline{u}}}\bigg[\partial_B\big(\frac{\Omega \overline{h}^{AB}}{r^{2\alpha-3|J|}}\big)\partial_A(u_\nu)_{ij,J}+\Omega\overline{ h}\frac{\Gamma_\frac{3}{2}\partial(u_\nu)_{ij,J}}{r^{2\alpha-3|J|}}\bigg]\partial_0(u_\nu)_{ij,J}d\mu_{^S\overline{g}}
\end{align}
It is immediate from the definition of the frame (\ref{Sframe}) and (\ref{dr/dtau}) that
\begin{align*}
\big|\partial_1(\frac{\Omega}{r^{2\alpha-3}})\big|\lesssim
\frac{\alpha}{r^{2\alpha-2}}&&
\partial_2(\frac{\Omega}{r^{\alpha-\frac{3}{2}}})
=\partial_3(\frac{\Omega}{r^{\alpha-\frac{3}{2}}})=0.
\end{align*}
Hence, similarly to (\ref{Pest1}) 
\begin{align}\label{Pest3}
&-\int_{\Sigma_\tau^{\overline{u}}}\bigg[\partial_B\big(\frac{\Omega \overline{h}^{AB}}{r^{2\alpha-3|J|}}\big)\partial_A(u_\nu)_{ij,J}+\Omega\overline{ h}\frac{\Gamma_\frac{3}{2}\partial(u_\nu)_{ij,J}}{r^{2\alpha-3|J|}}\bigg]\partial_0(u_\nu)_{ij,J}d\mu_{^S\overline{g}}\\
\tag{$|J|\leq2$}&\lesssim\;(\mathcal{E}_0^\frac{1}{2}+\alpha^2)
E_{3,\alpha}[{ u}]+E_{3,\alpha+1}[{ u}].
\end{align}
{\it Remark}: The term in the RHS of the preceding estimate with coefficient $\alpha^2$ is not critical. This is very important otherwise the overall estimates would not close, since the critical term with favourable sign in (\ref{intP/dtau}) is only of magnitude $\alpha$. 

We proceed to the boundary term in the RHS of (\ref{Pest2}). Recall that $\rho$ is constant on $\partial\Sigma_\tau^{\overline{u}}$ (\ref{rhoconst}), and decreasing in the interior direction of $\Sigma_\tau^{\overline{u}}$. Hence, the outward unit normal $N$ is the Schwarzschild normalized gradient of $\rho$ on $\Sigma^{\overline{u}}_\tau$, $N=\frac{^S\overline{\nabla}\rho}{|^S\overline{\nabla}\rho|}$.
Since $(\overline{h}^{AB})_{A,B=1,2,3}$ is a symmetric positive definite matrix,
the following standard inequality holds:
\begin{align*}
\bigg|\overline{h}^{AB}\frac{\partial_A(u_\nu)_{ij,J}}{r^{\alpha-\frac{3}{2}|J|}}
\Omega N_B\bigg|^2\leq&\;
\bigg(\overline{h}^{AB}\frac{\partial_A(u_\nu)_{ij,J}}{r^{\alpha-\frac{3}{2}|J|}}
\frac{\partial_B(u_\nu)_{ij,J}}{r^{\alpha-\frac{3}{2}|J|}}\bigg)
\big(\Omega^2\overline{h}^{AB}N_AN_B\big)\\
=&\;\bigg(\overline{h}^{AB}\frac{\partial_A(u_\nu)_{ij,J}}{r^{\alpha-\frac{3}{2}|J|}}
\frac{\partial_B(u_\nu)_{ij,J}}{r^{\alpha-\frac{3}{2}|J|}}\bigg)
\frac{\Omega^2\overline{h}^{AB}\partial_A(\rho)\partial_B(\rho)}{|^S\overline{\nabla}\rho|^2}\\
\tag{by (\ref{rhoeik})}
=&\;\bigg(\overline{h}^{AB}\frac{\partial_A(u_\nu)_{ij,J}}{r^{\alpha-\frac{3}{2}|J|}}
\frac{\partial_B(u_\nu)_{ij,J}}{r^{\alpha-\frac{3}{2}|J|}}\bigg)
\frac{-\overline{h}^{00}-2\Omega \overline{h}^{A0}\partial_A(\rho)}
{|^S\overline{\nabla}\rho|^2}
\end{align*}
Therefore, we have the bound
\begin{align}\label{Pest4}
&\int_{\partial\Sigma_\tau^{\overline{u}}}\Omega
\overline{h}^{AB}\frac{\partial_A(u_\nu)_{ij,J}}{r^{\alpha-\frac{3}{2}|J|}}
\frac{\partial_0(u_\nu)_{ij,J}}{r^{\alpha-\frac{3}{2}|J|}}
N_B dS\\
\notag\leq&\;\int_{\partial\Sigma_\tau^{\overline{u}}}\big|
\frac{\partial_0(u_\nu)_{ij,J}}{r^{\alpha-\frac{3}{2}|J|}}\big|
\sqrt{\frac{-\overline{h}^{00}-2\Omega \overline{h}^{A0}\partial_A(\rho)}
{|^S\overline{\nabla}\rho|}}
\sqrt{\frac{\overline{h}^{AB}}{|^S\overline{\nabla}\rho|}\frac{\partial_A(u_\nu)_{ij,J}}{r^{\alpha-\frac{3}{2}|J|}}
\frac{\partial_B(u_\nu)_{ij,J}}{r^{\alpha-\frac{3}{2}|J|}}}dS\\
\notag\leq&\;\frac{1}{2}\int_{\partial\Sigma_\tau^{\overline{u}}}
-\frac{\overline{h}^{00}}{|^S\overline{\nabla}\rho_\tau|}\frac{\big[\partial_0(u_\nu)_{ij,J}\big]^2}{r^{2\alpha-3|J|}}
-\frac{2\Omega \overline{h}^{A0}\partial_A(\rho)}{|^S\overline{\nabla}\rho_\tau|}
\frac{\big[\partial_0(u_\nu)_{ij,J}\big]^2}{r^{2\alpha-3|J|}}dS\\
&\notag+\frac{1}{2}\int_{\partial\Sigma_\tau^{\overline{u}}}
\frac{\overline{h}^{AB}}{|^S\overline{\nabla}\rho|}\frac{\partial_A(u_\nu)_{ij,J}}{r^{\alpha-\frac{3}{2}|J|}}
\frac{\partial_B(u_\nu)_{ij,J}}{r^{\alpha-\frac{3}{2}|J|}}dS
\end{align}
The remaining term to be estimated is the one on first line in the RHS of (\ref{Pest2}), which we rewrite 
\begin{align}\label{Pest5}
\notag&-\int_{\Sigma_\tau^{\overline{u}}}\Omega\bigg[
\overline{h}^{00}\frac{\partial_0^2(u_\nu)_{ij,J}}{r^{\alpha-\frac{3}{2}|J|}}
+\overline{h}^{AB}\frac{\partial_B\partial_A(u_\nu)_{ij,J}}{r^{\alpha-\frac{3}{2}|J|}}
\bigg]\frac{\partial_0(u_\nu)_{ij,J}}{r^{\alpha-\frac{3}{2}|J|}}d\mu_{^S\overline{g}}\\
=&-\int_{\Sigma_\tau^{\overline{u}}}\big(\overline{h}^{ab}\partial_a\partial_b(u_\nu)_{ij,J}\big)
\Omega\frac{\partial_0(u_\nu)_{ij,J}}{r^{2\alpha-3|J|}}d\mu_{^S\overline{g}}\\
\notag&+\int_{\Sigma_\tau^{\overline{u}}}2\Omega \overline{h}^{A0}
\frac{\partial_A\partial_0(u_\nu)_{ij,J}\partial_0(u_\nu)_{ij,J}}{r^{2\alpha-3|J|}}
+\Omega { h}
\frac{\Gamma_\frac{3}{2}\partial(u_\nu)_{ij,J}\partial_0(u_\nu)_{ij,J}}{r^{2\alpha-3|J|}}d\mu_{^S\overline{g}}
\end{align}
By taking the $\partial^{(J)}$ derivative ($J$ spatial multi-index $|J|\leq2$) of the first equation in (\ref{modboxu}) and commuting the differentiation in the LHS we obtain the equation
\begin{align}\label{modboxu,J}
 \notag&\overline{h}^{ab}\partial_a\partial_b(u_\nu)_{ij,J}\\
=&\;\partial^{(J)}\bigg[
\overline{O}\Gamma_{\frac{3}{2}}\partial u+\overline{O}\Gamma_3 u
+\overline{O}\Gamma_{\frac{9}{2}}(O-I)
+\overline{O}\Gamma_3\partial (O-I)\\
\notag&+\Gamma_3\overline{u}^2
+\overline{O}\overline{u}\partial \overline{u}
+\overline{u}^3
+\overline{O}\partial (\overline{O}-I)\partial \overline{u}\bigg]
+[\overline{h}^{ab}\partial_a\partial_b,\partial^{(J)}](u_\nu)_{ij},
\end{align}
where the commutator can in turn be written schematically as: [recall (\ref{overasym1}),(\ref{Gamma_q})]
\begin{align}\label{comm}
[\overline{h}^{ab}\partial_a\partial_b,\partial^{(J)}](u_\nu)_{ij}=&\;{ \partial}^2({ \overline{h}}){ \partial}^2(u_\nu)_{ij}
+\big[\Gamma_\frac{3}{2}\partial(\overline{h})+\Gamma_3\overline{h}\big]{ \partial}^2(u_\nu)_{ij}\\
&\notag+{ \partial(\overline{h})}{ \partial}^3(u_\nu)_{ij}
+{ \overline{h}\Gamma_\frac{3}{2}}{ \partial}^3(u_\nu)_{ij}\\
\notag&+\big[{ \partial(\overline{h})\Gamma_3}+\overline{h}\Gamma_\frac{9}{2}\big]{ \partial}(u_\nu)_{ij}\qquad\text{ if $|J|=2$}\\
\notag[\overline{h}^{ab}\partial_a\partial_b,\partial^{(J)}](u_\nu)_{ij}=&\; \partial({ \overline{h}}){ \partial}^2(u_\nu)_{ij}
+{ \overline{h}\Gamma_\frac{3}{2}}{ \partial}^2(u_\nu)_{ij}\qquad\qquad\text{if $|J|=1$}\\
\notag&+{ \overline{h}\Gamma_3}{ \partial}(u_\nu)_{ij}
\end{align}
We integrate by parts in the second term on the RHS of (\ref{Pest5}) and argue 
similarly to (\ref{Pest3}) to get 
\begin{align}\label{Pest6}
\notag&\int_{\Sigma_\tau^{\overline{u}}}2\Omega \overline{h}^{A0}
\frac{\partial_A\partial_0(u_\nu)_{ij,J}\partial_0(u_\nu)_{ij,J}}{r^{2\alpha-3|J|}}
+\Omega\overline{h}
\frac{\Gamma_\frac{3}{2}\partial(u_\nu)_{ij,J}\partial_0(u_\nu)_{ij,J}}{r^{2\alpha-3|J|}}d\mu_{^S\overline{g}}\\
=&\int_{\partial\Sigma_\tau^{\overline{u}}}\Omega \overline{h}^{A0}
\frac{\big[\partial_0(u_\nu)_{ij,J}\big]^2}{r^{2\alpha-3|J|}}N_a dS
-\int_{\Sigma_\tau^{\overline{u}}}\Omega \partial_A(\overline{h}^{A0})
\frac{\big[\partial_0(u_\nu)_{ij,J}\big]^2}{r^{2\alpha-3|J|}}d\mu_{^S\overline{g}}\\
&\notag-\int_{\Sigma_\tau^{\overline{u}}}\overline{h}^{A0}
\partial_A(\frac{\Omega}{r^{2\alpha-3|J|}})\big[\partial_0(u_\nu)_{ij,J}\big]^2d\mu_{^S\overline{g}}
+\int_{\Sigma_\tau^{\overline{u}}}\Omega \overline{h}\Gamma_\frac{3}{2}
\frac{\big[\partial_0(u_\nu)_{ij,J}\big]^2}{r^{2\alpha-3|J|}}d\mu_{^S\overline{g}}\\
&\notag+\int_{\Sigma_\tau^{\overline{u}}}\Omega \overline{h}
\frac{\Gamma_\frac{3}{2}\partial(u_\nu)_{ij,J}\partial_0(u_\nu)_{ij,J}}{r^{2\alpha-3|J|}}d\mu_{^S\overline{g}}\\
\notag\leq&\int_{\partial\Sigma_\tau^{\overline{u}}}\Omega \overline{h}^{A0}
\frac{\big[\partial_0(u_\nu)_{ij,J}\big]^2}{r^{2\alpha-3|J|}}N_A dS
+C(\mathcal{E}_0^{\frac{1}{2}}+\alpha^2)E_{3,\alpha}[{ u}]
+CE_{3,\alpha+1}[{ u}]
\end{align}
Finally, for the last and main term in the first line of the RHS of (\ref{Pest5}) we recall that $|\Omega|\lesssim\frac{1}{r^\frac{1}{2}}$ to obtain directly from Cauchy's inequality 
\begin{align}\label{Pest7}
\notag&-\int_{\Sigma_\tau^{\overline{u}}}\big(\overline{h}^{ab}\partial_a\partial_b(u_\nu)_{ij,J}\big)\Omega\frac{\partial_0(u_\nu)_{ij,J}}{r^{2\alpha-3|J|}}d\mu_{^S\overline{g}}\\
\lesssim&\;
\big\|\frac{\overline{h}^{ab}\partial_a\partial_b(u_\nu)_{ij,J}}{r^{\alpha-\frac{3}{2}|J|-\frac{1}{2}}}\big\|^2_{L^2}
+\big\|\frac{\partial_0(u_\nu)_{ij,J}}{r^{\alpha-\frac{3}{2}|J|+1}}\big\|^2_{L^2}\\
\notag\lesssim&\;\big\|\frac{\overline{h}^{ab}\partial_a\partial_b(u_\nu)_{ij,J}}{r^{\alpha-\frac{3}{2}|J|-\frac{1}{2}}}\big\|^2_{L^2}
+\|\partial_0(u_\nu)_{ij}\|^2_{H^{2,\alpha-\frac{1}{2}}}
\end{align}
We proceed by plugging the RHS of (\ref{modboxu,J}) into the first term in the last inequality (\ref{Pest7}) above and treat each arising group of terms separately. Employing the basic inequalities in Proposition \ref{adaptineq} along with the bounds of ${ \overline{O}},{ \partial(\overline{O})},{ \overline{u}}$ (\ref{venest}), (\ref{norme_0barf}) and (\ref{Linftybarf}) we derive:
\begin{align}\label{Pest8}
&\big\|\frac{\partial^{(J)}\big[\overline{O}\Gamma_{\frac{3}{2}}\partial u+\overline{O}\Gamma_3 u\big]}{r^{\alpha-\frac{3}{2}|J|-\frac{1}{2}}}\big\|^2_{L^2}\\
\notag\lesssim&\;\|{ \overline{O}}\|^2_{L^\infty}E_{3,\alpha+1}[{ u}]
+\|{ \partial(\overline{O})}\|^2_{L^\infty}E_{2,\alpha}[{ u}]+\|{ u}\|^2_{H^{2,\alpha+1}}\\
&+\|\frac{\partial^{(J)}({ \overline{O}}) \Gamma_\frac{3}{2}\partial u}{r^{\alpha-\frac{3}{2}|J|-\frac{1}{2}}}\|^2_{L^2}
\notag+\|\frac{\partial^{(J)}({ \overline{O}}){ \Gamma_3 u}}{r^{\alpha-\frac{3}{2}|J|-\frac{1}{2}}}\|^2_{L^2}\\
\tag{the last two terms appear only in the case $|J|=2$}\lesssim&\;\mathcal{E}({ \overline{u},\overline{O}};\alpha,T)E_{3,\alpha}[{ u}]+E_{3,\alpha+1}[{ u}]
+\big(\|{ u}\|^2_{L^\infty}+\|\partial u\|^2_{L^\infty}\big)\|\frac{\partial^{(J)}({ \overline{O}})}{r^{\alpha-\frac{1}{2}}}\|_{L^2}^2\\
\lesssim&\;\notag\mathcal{E}_0E_{3,\alpha}[{ u}]+E_{3,\alpha+1}[{ u}]
\end{align}
\begin{align}\label{Pest9}
&\big\|\frac{\partial^{(J)}\big[\overline{O}\Gamma_{\frac{9}{2}}(O-I)
+\overline{O}\Gamma_3\partial (O-I)\big]}{r^{\alpha-\frac{3}{2}|J|-\frac{1}{2}}}\big\|^2_{L^2}\\
\notag\lesssim&\;\|{ O-I}\|^2_{H^{3,\alpha+\frac{5}{2}}}
+\|{ \partial\overline{O}}\|^2_{L^\infty}\|{ O-I}\|^2_{H^{2,\alpha+\frac{3}{2}}}\\
&\notag+\big(\|\frac{{ O-I}}{r^\frac{3}{2}}\|^2_{L^\infty}+\|{ \partial(O-I)}\|_{L^\infty}^2\big)\|\frac{\partial^{(J)}({ \overline{O}})}{r^{\alpha-\frac{1}{2}}}\|^2_{L^2}\\
\tag{we include the last term only when $|J|=2$ and utilize (\ref{norme_0barf}),(\ref{norme_0f})}
\lesssim&\;\|{ O-I}\|^2_{H^{3,\alpha+\frac{5}{2}}}
+(\mathcal{E}_0^2+\mathcal{E}_0)\|{ O-I}\|^2_{H^{2,\alpha+\frac{3}{2}}}
+\mathcal{E}_0\|O-I\|^2_{H^{3,\alpha+\frac{3}{2}}}\\
\notag&+\mathcal{E}_0^3+\mathcal{E}_0\|u\|^2_{H^{2,\alpha}}
\end{align}
\begin{align}\label{Pest10}
&\big\|\frac{\partial^{(J)}\big[\Gamma_3\overline{u}^2
+\overline{O}\overline{u}\partial \overline{u}
+\overline{u}^3
+\overline{O}\partial (\overline{O}-I)\partial\overline{u}\big]}{r^{\alpha-\frac{3}{2}|J|-\frac{1}{2}}}\big\|^2_{L^2}\\
\notag\lesssim&\;\bigg[\|\frac{{ \overline{u}}}{r^\frac{3}{2}}\|^2_{L^\infty}+\|\frac{{ \partial\overline{u}}}{r^\frac{3}{2}}\|^2_{L^\infty}
+\|{ \partial\overline{O}}\|^2_{L^\infty}\big(\|{ \overline{u}}\|^2_{L^\infty}+\|\partial\overline{u}\|^2_{L^\infty}\big)\bigg]E_{3,\alpha}[{ \overline{u}}]\\
&\notag+\|{ \overline{u}}\|^2_{L^\infty}\|\partial\overline{u}\|^2_{L^\infty}\|{ \overline{O}-I}\|^2_{H^{2,\alpha+\frac{3}{2}}}
+\big(\|{ \overline{u}}\|^4_{L^\infty}
+\|{ \overline{u}}\|^2_{L^\infty}\|\partial\overline{u}\|^2_{L^\infty}\big)
\|{ \overline{u}}\|^2_{H^{2,\alpha}}\\
&\notag+\|\partial\overline{O}\|^2_{L^\infty}E_{3,\alpha}[{ \overline{u}}]
+\|\frac{\partial^2( \overline{O}-I)\partial^2 \overline{u}}{r^{\alpha-3-\frac{1}{2}}}\|^2_{L^2}+\|\partial\overline{u}\|^2_{L^\infty}
\|{ \overline{O}-I}\|^2_{H^{3,\alpha+\frac{3}{2}}}\\
\notag&+\big(\|{ \partial\overline{O}}\|^4_{L^\infty}+\|{ \partial\overline{O}}\|^2_{L^\infty}\|{ \partial\overline{u}}\|^2_{L^\infty}\big)\big(\|{ \overline{O}-I}\|^2_{H^{2,\alpha+\frac{3}{2}}}+E_{2,\alpha}[{ \overline{u}}]\big)\\
\lesssim&\;\notag\mathcal{E}({ \overline{u},\overline{O}};\alpha,T)^2+\mathcal{E}({ \overline{u},\overline{O}};\alpha,T)^3
+\|\frac{\partial^2({ \overline{O}})}{r^{\frac{\alpha}{2}-\frac{3}{2}-\frac{1}{4}}}\|^2_{L^4}
\|\frac{\partial^2({ \overline{u}})}{r^{\frac{\alpha}{2}-\frac{3}{2}-\frac{1}{4}}}\|^2_{L^4}\\
\tag{employing the $L^4$ estimate (\ref{L4})}\lesssim&\;\mathcal{E}_0^2+\mathcal{E}_0^3+\alpha^3\mathcal{E}_0^2
\end{align}
By (\ref{Pest8})-(\ref{Pest10}) we have the following lemma.
\begin{lemma}\label{e_0e_0u}
$\partial_0^2(u_\nu)_{ij}\in C([0,T];H^{1,\alpha-3})\cap L^2([0,T];H^{1,\alpha-2})$ and moreover the following estimate holds:
\begin{align}\label{e_0e_0uest}
\notag&\big\|\frac{\partial^{(J)}\partial_0^2(u_\nu)_{ij}}{r^{\alpha-\frac{3}{2}|J|-\frac{1}{2}}}\big\|^2_{L^2[\tau]}\\
\lesssim&\; \mathcal{E}_0\big(E_{3,\alpha}[{ u}]+\|{ O-I}\|^2_{H^{3,\alpha+\frac{3}{2}}}\big)+E_{3,\alpha+1}[{ u}]
+\|{ O-I}\|^2_{H^{3,\alpha+\frac{5}{2}}}\\
\notag&+\alpha^3\mathcal{E}_0^2
+\mathcal{E}_0^3,
\end{align}
for $|J|\leq1$, $J\subset\{1,2,3\}$, $\tau\in(0,T)$.
\end{lemma}
\begin{proof}
The proof follows by solving for $\partial_0^2(u_\nu)_{ij}$ in the equation (\ref{modboxu}) and summing up the above estimates (\ref{Pest8})-(\ref{Pest10}).
\end{proof}
To bound the commutator (\ref{comm}) we treat the cases $|J|=2$, $|J|=1$ separately. For $|J|=1$:
\begin{align}\label{Pest11}
&\big\|\frac{[\overline{h}^{ab}\partial_a\partial_b,\partial^{(J)}](u_\nu)_{ij}}{r^{\alpha-\frac{3}{2}-\frac{1}{2}}}\big\|^2_{L^2}\\
\notag=&\;\big\|\frac{\partial({ \overline{h}}){ \partial}^2(u_\nu)_{ij}
+{ \overline{h}\Gamma_\frac{3}{2}}{ \partial}^2(u_\nu)_{ij} 
+{ \overline{h}\Gamma_3}{ \partial}(u_\nu)_{ij}}{r^{\alpha-\frac{3}{2}-\frac{1}{2}}}\big\|^2_{L^2}\\
\notag\lesssim&\;\|\partial\overline{ h}\|^2_{L^\infty}\|\frac{\partial^2(u_\nu)_{ij}}{r^{\alpha-\frac{3}{2}-\frac{1}{2}}}\|^2_{L^2}+
\|\overline{ h}\|^2_{L^\infty}\|\frac{\partial^2(u_\nu)_{ij}}{r^{\alpha-\frac{1}{2}}}\|^2_{L^2}
+\|\overline{ h}\|^2_{L^\infty}\|\frac{\partial(u_\nu)_{ij}}{r^{\alpha+1}}\|^2_{L^2}\\
\tag{employing (\ref{e_0e_0uest}) in the case $\partial^2(u_\nu)_{ij}=\partial_0^2(u_\nu)_{ij}$}\lesssim&\;\mathcal{E}_0\big(E_{3,\alpha}[{ u}]+\|{ O-I}\|^2_{H^{3,\alpha+\frac{3}{2}}}\big)+E_{3,\alpha+1}[{ u}]
+\|{ O-I}\|^2_{H^{3,\alpha+\frac{5}{2}}}\\
&\notag+\alpha^3\mathcal{E}_0^2
+\mathcal{E}_0^3
\end{align}
When $|J|=2$ we have
\begin{align}\label{Pest12}
&\big\|\frac{[\overline{h}^{ab}\partial_a\partial_b,\partial^{(J)}](u_\nu)_{ij}}{r^{\alpha-3-\frac{1}{2}}}\big\|^2_{L^2}\\
\notag\lesssim&\;
\big\|\frac{{ \partial}^2({ \overline{h}}){ \partial}^2(u_\nu)_{ij}
+\big[\Gamma_\frac{3}{2}\partial(\overline{h})+\Gamma_3\overline{h}\big]{ \partial}^2(u_\nu)_{ij}}{r^{\alpha-3-\frac{1}{2}}}\big\|^2_{L^2}\\
\notag&+\big\|\frac{{ \partial(\overline{h})}{ \partial}^3(u_\nu)_{ij}
+{ \overline{h}\Gamma_\frac{3}{2}}{ \partial}^3(u_\nu)_{ij}
+\big[{ \partial(\overline{h})\Gamma_3}+\overline{h}\Gamma_\frac{9}{2}\big]{ \partial}(u_\nu)_{ij}}{r^{\alpha-3-\frac{1}{2}}}\big\|^2_{L^2}\\
\tag{note that term $\partial^3(u_\nu)_{ij}$ contains at most two $\partial_0$ derivatives}\lesssim&\;\|\frac{\partial^2 \overline{h}}{r^{\frac{\alpha}{2}-\frac{3}{2}-\frac{1}{4}}}\|^2_{L^4}
\|\frac{\partial^2(u_\nu)_{ij}}{r^{\frac{\alpha}{2}-\frac{3}{2}-\frac{1}{4}}}\|^2_{L^4}
+\|\partial \overline{h}\|^2_{L^2}\big(E_{3,\alpha}[{ u}]
+\|\partial_0^2(u_\nu)_{ij}\big\|^2_{H^{1,\alpha-3}}\big)\\
\notag&+\|\overline{ h}\|^2_{L^\infty}\big(E_{3,\alpha+1}[{ u}]+\big\|\partial_0^2(u_\nu)_{ij}\|^2_{H^{1,\alpha-2}}\big)\\
\tag{employing the $L^4$ estimate (\ref{L4}) and (\ref{e_0e_0uest})}\lesssim&\;\alpha^3\mathcal{E}_0E_{3,\alpha}[{ u}]+\mathcal{E}_0E_{3,\alpha}[{ u}]+E_{3,\alpha+1}[{ u}]
+\mathcal{E}_0\|{ O-I}\|^2_{H^{3,\alpha+\frac{3}{2}}}
+\|{ O-I}\|^2_{H^{3,\alpha+\frac{5}{2}}}\\
\notag&+\alpha^3\mathcal{E}_0^2
+\mathcal{E}_0^3
\end{align}
{\bf Summary}: Incorporating (\ref{intP/dtau})-(\ref{Pest12})
in (\ref{E/dtau}) we conclude that
\begin{align}\label{sumEest}
\notag&\partial_\tau\int_{\Sigma_\tau^{\overline{u}}}P_{J,\alpha}d\mu_{^S\overline{g}}
+8M^2e^{-1}(1-\tau)\alpha \int_{\Sigma_\tau^{\overline{u}}}P_{J,\alpha+1}d\mu_{^S\overline{g}}\\
\lesssim&\;
(\mathcal{E}_0^\frac{1}{2}+\mathcal{E}_0+\alpha^2+\alpha^3\mathcal{E}_0)E_{3,\alpha}[{ u}]+E_{3,\alpha+1}[{ u}]+\|{ O-I}\|^2_{H^{3,\alpha+\frac{5}{2}}}\\
&\notag+\mathcal{E}_0\|{ O-I}\|^2_{H^{3,\alpha+\frac{3}{2}}}
+\alpha^3\mathcal{E}_0^2+\mathcal{E}_0^3
\end{align}
Summing over the indices $\nu,i,j$ and $J$, $|J|\leq2$, we arrive at the desired estimate (\ref{E_sest}).
\end{proof}
\begin{proof}[Proof of (\ref{Oest})]
Let $J$, $|J|\leq3$, be a spatial multi-index. Like in the case of (\ref{E/dtau}), it follows from the coarea formula and the asymptotics (\ref{dr/dtau}),(\ref{dmu/dtau}) that 
\begin{align}\label{f-I/dtau}
&\frac{1}{2}\partial_\tau\|\frac{O^d_{c,J}-{{I_c}^d}_{,J} }{r^{\alpha+3-\frac{3}{2}|J|}}\|^2_{L^2(\Sigma_\tau^{\overline{u}})}
=
-\frac{1}{2}\int_{\partial\Sigma_\tau^{\overline{u}}}\frac{(O^d_{c,J}-{{I_c}^d}_{,J})^2}{|^S\overline{\nabla}\rho|r^{2\alpha+6-3|J|}}dS\\
&\notag-(\alpha+3-\frac{3|J|}{2})\int_{\Sigma_\tau^{\overline{u}}}\frac{(O^d_{c,J}-{{I_c}^d}_{,J})^2}{r^{2\alpha+7-3|J|}}\partial_\tau rd\mu_{^S\overline{g}}\\
&\notag+\int_{\Sigma_\tau^{\overline{u}}}\Omega\frac{(O^d_{c,J}-{{I_c}^d}_{,J})\partial_0(O^d_{c,J})}{r^{2\alpha+6-3|J|}}d\mu_{^S\overline{g}}
+\frac{1}{2}\int_{\Sigma_\tau^{\overline{u}}}\frac{(O^d_{c,J}-{{I_c}^d}_{,J})^2}{r^{2\alpha+6-3|J|}}\partial_\tau d\mu_{^S\overline{g}}\\
\notag\leq&\;-4M^2e^{-1}\alpha(1-\tau)\|\frac{O^d_{c,J}-{{I_c}^d}_{,J}}{r^{\alpha+4-\frac{3}{2}|J|}}\|^2_{L^2(\Sigma_\tau^{\overline{u}})}\\
&+\int_{\Sigma_\tau^{\overline{u}}}\Omega\frac{(O^d_{c,J}-{{I_c}^d}_{,J})\partial_0(O^d_{c,J})}{r^{2\alpha+6-3|J|}}d\mu_{^S\overline{g}}
\notag+C\|\frac{O^d_{c,J}-{{I_c}^d}_{,J}}{r^{\alpha+4-\frac{3}{2}|J|}}\|^2_{L^2(\Sigma_\tau^{\overline{u}})},
\end{align}
where $O_{c,J}^d-{{I_c}^d}_{,J}:=\partial^{(J)}(O^d_c-{I_c}^d)$. By Cauchy's inequality we have ($\Omega\lesssim\frac{1}{r^\frac{1}{2}}$)
\begin{align}\label{f-Iest}
\int_{\Sigma_\tau^{\overline{u}}}\Omega\frac{(O^d_{c,J}-{{I_c}^d}_{,J})\partial_0(O^d_{c,J})}{r^{2\alpha+6-3|J|}}d\mu_{^S\overline{g}}\lesssim\|\frac{O^d_{c,J}-{{I_c}^d}_{,J}}{r^{\alpha+4-\frac{3}{2}|J|}}\|^2_{L^2}+\|\frac{\partial_0(O^d_{c,J})}{r^{\alpha+\frac{5}{2}-\frac{3}{2}|J|}}\|^2_{L^2}
\end{align}
Taking the $\partial^{(J)}$ derivative of the ODE in (\ref{modboxu}) we obtain 
\begin{align}\label{mode_0f,J}
 \partial_0(O^d_{c,J}-{{I_c}^d}_{,J})=&\;\partial^{(J)}\bigg[\Gamma_{\frac{3}{2}}(O-I)+(\overline{O}-I)\overline{u}
+u\bigg]+[\partial_0,\partial^{(J)}](O^d_c-{{I_c}^d})
\end{align}
The commutator in the RHS of (\ref{mode_0f,J}) schematically reads 
\begin{align}\label{commf}
[\partial_0,\partial^{(J)}](O^d_c-{{I_c}^d})=&\;\Gamma_\frac{3}{2}\partial(O^d_c-{{I_c}^d})&&\text{if $|J|=1$}\\
\notag=&\;\Gamma_3\partial(O^d_c-{{I_c}^d})
+\Gamma_\frac{3}{2}\partial^2(O^d_c-{{I_c}^d})&&\text{if $|J|=2$}\\
\notag=&\;\Gamma_\frac{9}{2}\partial(O^d_c-{{I_c}^d})
+\Gamma_3\partial^2(O^d_c-{{I_c}^d})&&\text{if $|J|=3$}\\
&+\notag \Gamma_\frac{3}{2}\partial^3(O^d_c-{{I_c}^d}),
\end{align}
where we note that at most one $\partial_0$ derivative of $O^d_c-{I_c}^d$ appears in the preceding expressions. 
Hence, we deduce directly from (\ref{mode_0f,J}):
\begin{align}\label{f-Iest2}
&\|\frac{\partial_0(O^d_{c,J}-{{I_c}^d}_{,J})}{r^{\alpha+\frac{5}{2}-\frac{3}{2}|J|}}\|^2_{L^2}\\
\notag\lesssim&\;\|{ \Gamma_\frac{3}{2}(O-I)}\|^2_{H^{3,\alpha+1}}+\|({ \overline{O}- I}){ \overline{u}}\|^2_{H^{3,\alpha+1}}+
+\|u\|^2_{H^{3,\alpha+1}}\\
&\notag+\|\frac{[\partial_0,\partial^{(J)}](O^d_c-{I_c}^d)}{r^{\alpha+\frac{5}{2}-\frac{3}{2}|J|}}\|^2_{L^2}\\
\tag{exmploying Lemma \ref{Sobprop} and applying the $L^\infty$ bound 
on $({ \overline{O}- I}){ \overline{u}}$}
\lesssim&\;
\|{ O-I}\|^2_{H^{3,\alpha+\frac{5}{2}}}+E_{3,\alpha+1}[{ u}]
+\|{ \overline{O}- I}\|^2_{H^{3,\alpha+\frac{3}{2}}}\|{ \overline{u}}\|^2_{H^{3,\alpha}}\\
&\notag+\|\frac{{ \Gamma_\frac{3}{2}\partial}(O^d_c-{I_c}^d)}{r^{\alpha+\frac{5}{2}-\frac{3}{2}}}\|^2_{L^2}
+\|\frac{{ \Gamma_3\partial}(O^d_c-{I_c}^d)
+\Gamma_\frac{3}{2}\partial^2(O^d_c-{I_c}^d)}{r^{\alpha+\frac{5}{2}-3}}\|^2_{L^2}\\
\notag&+\|\frac{\Gamma_\frac{9}{2}\partial(O^d_c-{I_c}^d)
+\Gamma_3\partial^2(O^d_c-{I_c}^d)
+\Gamma_\frac{3}{2}\partial^3(O^d_c-{I_c}^d)}{r^{\alpha+\frac{5}{2}-\frac{9}{2}}}\|^2_{L^2}\\
\lesssim&\notag\;
\|{ O-I}\|^2_{H^{3,\alpha+\frac{5}{2}}}
+\mathcal{E}({ \overline{u},\overline{O}};\alpha,T)^2
+E_{3,\alpha+1}[{ u}]
\end{align}
Combining (\ref{f-I/dtau})-(\ref{f-Iest2}) we derive
\begin{align}\label{f-Iest3}
&\frac{1}{2}\partial_\tau\|\frac{O^d_{c,J}-{{I_c}^d}_{,J}}{r^{\alpha+3-\frac{3}{2}|J|}}\|^2_{L^2(\Sigma_\tau^{\overline{u}})}+4M^2e^{-1}\alpha(1-\tau)\|\frac{O^d_{c,J}-{{I_c}^d}_{,J}}{r^{\alpha+4-\frac{3}{2}|J|}}\|^2_{L^2(\Sigma_\tau^{\overline{u}})}\\
\notag\lesssim&\;\|{ O-I}\|^2_{H^{3,\alpha+\frac{5}{2}}}+E_{3,\alpha+1}[{ u}]+\mathcal{E}_0^2
\end{align}
Taking into account the set of indices $c,d$ and $J$, $|J|\leq3$, we complete the proof of (\ref{Oest}) and hence of Proposition \ref{mainenest}.
\end{proof}
\subsection{Contraction mapping in $H^{2,\alpha}$}\label{Contr}

We proceed to show that the mapping defined via (\ref{modboxu}) in the beginning of \S\ref{Enest} is a contraction.
Let us consider another solution set 
$(\tilde{u}_\nu)_{ij},\tilde{O}^d_c,{ \tilde{\overline{u}},\tilde{\overline{O}}}$ 
of the analogous coupled system (\ref{modboxu}). Setting 
\begin{align}\label{du,df}
(du_\nu)_{ij}=(u_\nu)_{ij}-(\tilde{u}_\nu)_{ij},\;\;d{ \overline{u}}={ \overline{u}}-{ \tilde{\overline{u}}},\;\;
dO^d_c=O^d_c-\tilde{O}^d_c,\;\;d{ \overline{O}}={ \overline{O}}-{ \tilde{\overline{O}}}
\end{align}
we obtain schematically the new system of equations (depicting only the types of terms in the RHS suppressing the particular indices) 
\begin{align}\label{boxdu}
&\notag 
\overline{h}^{ab}\partial_a\partial_b(du_\nu)_{ij}\\
=&\;
{ \overline{O}}\Gamma_\frac{3}{2}\partial(d{ u})
+{ \overline{O}}\Gamma_3d{ u}
+\overline{O}\Gamma_\frac{9}{2}d{ O}\\
&\notag+d{ \overline{O}}\big[\Gamma_\frac{3}{2}\partial\tilde{u}+\Gamma_3{ \tilde{u}}+\Gamma_\frac{9}{2}({ \tilde{O}-I})
+\Gamma_3\partial{ \tilde{O}}\big]\\
\notag&+{ \overline{O}}\Gamma_3\partial(d{ f})
+({ \overline{O}+\tilde{\overline{O}}})d{ \overline{O}}\partial^2(\tilde{u}_\nu)_{ij}+G(d{ \overline{u}},d{ \overline{O}}),
\end{align}
where
\begin{align}\label{G}
G(d{ \overline{u}},d{ \overline{O}})=&\;\notag \Gamma_\frac{3}{2}d{ \overline{u}}({ \overline{u}}+{ \tilde{\overline{u}}}+{ \overline{u}}^2+{\tilde{\overline{u}}}^2+{ \overline{u}\tilde{\overline{u}}})
+{ \overline{O}}{ \overline{u}}\partial(d{ \overline{u}})\\
&+{ \overline{O}}d{ \overline{u}}\partial{ \tilde{\overline{u}}}
+d{ \overline{O}}{ \tilde{\overline{u}}}\partial{ \tilde{\overline{u}}}
+{ \overline{O}}\partial({ \overline{O}}){ \partial}(d{ \overline{u}})\\
&\notag+{ \overline{O}}\partial(d{ \overline{O}})\partial{ \tilde{\overline{u}}}
+d{ \overline{O}}\partial({ \tilde{\overline{O}}})\partial{ \tilde{\overline{u}}}
\end{align}
and
\begin{align}\label{e_0df}
\partial_0(dO^d_c)=\Gamma_\frac{3}{2}dO
+(\overline{O}-I)d{ \overline{u}}
+\tilde{\overline{u}}d{ \overline{O}}+du
\end{align}
Further, we assume that both sets of variables we have introduced are consistent with the energy estimate (\ref{uenest}) we have established in the previous
 subsection:
\begin{align}\label{tildeenest}
\mathcal{E}({ u,O};\alpha,T),\;
\mathcal{E}({ \overline{u},\overline{O}};\alpha,T),\;
\mathcal{E}({ \tilde{u},\tilde{O}};\alpha,T),\;
\mathcal{E}({ \tilde{\overline{u}},\tilde{\overline{O}}};\alpha,T)\leq 2\mathcal{E}_0.
\end{align}
{\it Claim}: For large enough $\alpha>0$ and $T>0$ is sufficiently small the following contraction holds:
\begin{align}\label{contr}
E_{2,\alpha}[d{ u}]+\sum_{c,d}\|dO^d_c\|^2_{H^{2,\alpha+\frac{3}{2}}}\leq \kappa\big(E_{2,\alpha}[d{ \overline{u}}]+\sum\|d{ \overline{O}}\|^2_{H^{2,\alpha+\frac{3}{2}}}\big),
\end{align}
for some $0<\kappa<1$.
\begin{remark}\label{contrH2}
We are forced to close the contraction mapping argument in $H^{2,\alpha}$, having one derivative less than the space of the energy estimate (\ref{tildeenest}), see \S\ref{Enest}, as it is common in 2$^{\text{nd}}$-order quasilinear hyperbolic PDE \cite{Choq2}, because of the problematic term $(\overline{O}+\tilde{\overline{O}})d{ \overline{O}}\partial^2(\tilde{u}_\nu)_{ij}$ in (\ref{boxdu}), which is generated from the difference of the top order terms in the LHS.
\end{remark}
\begin{proposition}\label{maincontrest}
Under the above considerations, the following estimates hold:
\begin{align}\label{contrEest}
&\partial_\tau E_{2,\alpha}[d{ u}]+8M^2e^{-1}(1-\tau)\alpha E_{2,\alpha+1}[d{ u}]\\
\notag\lesssim&\; (\mathcal{E}_0^\frac{1}{2}+\mathcal{E}_0+\alpha^2)E_{2,\alpha}[d{ u}]+E_{2,\alpha+1}[d{ u}]
+\|d{ f}\|^2_{H^{2,\alpha+\frac{5}{2}}}\\
&\notag+(\mathcal{E}_0+\mathcal{E}_0^2+\alpha^3\mathcal{E}_0)\big(E_{2,\alpha}[d{ \overline{u}}]+\|d{ \overline{O}}\|^2_{H^{2,\alpha+\frac{3}{2}}}\big)
\end{align}
\begin{align}\label{contrfest}
&\frac{1}{2}\partial_\tau\sum_{c,d}\|dO^d_c\|^2_{H^{2,\alpha+\frac{3}{2}}}+4M^2e^{-1}(1-\tau)\alpha\sum_{c,d}\|dO^d_c\|^2_{H^{2,\alpha+\frac{5}{2}}}\\
\lesssim&\;\notag\|d{ O}\|^2_{H^{2,\alpha+\frac{5}{2}}}+E_{2,\alpha+1}[d{ u}]+\mathcal{E}_0\big(E_{2,\alpha}[d{ \overline{u}}]+\|d{ \overline{O}}\|^2_{H^{2,\alpha+\frac{3}{2}}}\big),
\end{align}
for all $\tau\in(0,T)$.
\end{proposition}
Assuming Proposition \ref{maincontrest} we prove the above claim (\ref{contr}).
After summing (\ref{contrEest}),(\ref{contrfest}), 
we absorb into the LHS the critical terms 
\begin{align*}
E_{2,\alpha+1}[d{ u}],\|dO\|^2_{H^{2,\alpha+\frac{5}{2}}},
\end{align*}
which appear in the RHS of the above inequalities. This is done by picking the parameter $\alpha$ sufficiently large (but finite). The contraction estimate (\ref{contr}) then follows from Gronwall's inequality for $T>0$ suitably small.
\begin{proof}[Proof of Proposition \ref{maincontrest}]
The proof follows exactly the lines of the proof of Proposition \ref{mainenest}. The only notable difference lies in the estimation of the analogous term to (\ref{Pest7}), derived in (\ref{Pest8})-(\ref{Pest12}). We sketch the argument in the present situation:\\
Let $J$ denote at most one spatial index, $|J|\leq1$, either $1,2$ or $3$. The main term to be estimated is
\begin{align}\label{mainest}
&-\int_{\Sigma_\tau^{\overline{u}}}\partial^{(J)}\bigg[\text{RHS of (\ref{boxdu})}\bigg]\Omega \frac{\partial_0(du_\nu)_{ij,J}}{r^{2\alpha-3}}d\mu_{^S\overline{g}}\\
\lesssim\tag{recall $\Omega\lesssim\frac{1}{r^\frac{1}{2}}$}&\;\|\frac{\partial_0(du_\nu)_{ij,J}}{r^{\alpha-\frac{1}{2}}}\|^2_{L^2}
+\big\|\frac{\partial^{(J)}\big[\text{RHS of (\ref{boxdu})}\big]}{r^{\alpha-2}}\big\|^2_{L^2},
\end{align}
where $(du_\nu)_{ij,J}:=\partial^{(J)}(du_\nu)_{ij}$. Plugging in (\ref{boxdu}) and using the basic estimates in Proposition \ref{adaptineq}, along with the assumption (\ref{tildeenest}) we obtain
\begin{align}\label{mainest2}
&\big\|\frac{\partial^{(J)}\big[\text{RHS of (\ref{boxdu})}\big]}{r^{\alpha-2}}\big\|^2_{L^2}\\
\lesssim&\;
\notag\big\|\frac{\partial^{(J)}\big({ \overline{O}}\Gamma_\frac{3}{2}\partial(d{ u})
+{ \overline{O}}\Gamma_3d{ u}
+\overline{O}\Gamma_\frac{9}{2}d{ O}\big)}{r^{\alpha-2}}\big\|^2_{L^2}\\
&\notag+\big\|\frac{\partial^{(J)}\big(d{ \overline{O}}\big[\Gamma_\frac{3}{2}\partial\tilde{u}+\Gamma_3{ \tilde{u}}+\Gamma_\frac{9}{2}({ \tilde{O}-I})
+\Gamma_3\partial{ \tilde{O}}\big]\big)}{r^{\alpha-2}}\big\|^2_{L^2}\\
\notag&+\big\|\frac{\partial^{(J)}\big({ \overline{O}}\Gamma_3\partial(dO)\big)}{r^{\alpha-2}}\big\|^2_{L^2}
+\big\|\frac{\partial^{(J)}\big[({ \overline{O}+\tilde{\overline{O}}})d{ \overline{O}}\partial^2(\tilde{u}_\nu)_{ij}\big]}{r^{\alpha-2}}\big\|^2_{L^2}\\
\notag&+\|\frac{\partial^{(J)}G(d{ \overline{u}},d{ \overline{O}})}{r^{\alpha-2}}\|^2_{L^2}\\
\tag{recall the asymptotics (\ref{Gamma_q})}
\lesssim&\;\|\frac{\partial^{(J)}\partial(d{ u})}{r^{\alpha-\frac{1}{2}}}\|^2_{L^2}
+\|\frac{\partial(d{ u})}{r^{\alpha+1}}\|^2_{L^2}
+\|\frac{d{ u}}{r^{\alpha+\frac{5}{2}}}\|^2_{L^2}
+\mathcal{E}_0\big(\|\frac{\partial(d{ u})}{r^{\alpha-\frac{1}{2}}}\|^2_{L^2}
+\|\frac{d{ u}}{r^{\alpha+1}}\|^2_{L^2}\big)
\\
\notag&+\|\frac{\partial^{(J)}(dO)}{r^{\alpha+\frac{5}{2}}}\|_{L^2}^2
+\|\frac{dO}{r^{\alpha+4}}\|_{L^2}^2
+\mathcal{E}_0\|\frac{dO}{r^{\alpha+\frac{5}{2}}}\|_{L^2}^2\\
\notag&+\mathcal{E}_0\|\frac{\partial^{(J)}(d{ \overline{O}})}{r^{\alpha+\frac{3}{2}}}\|^2_{L^2}+\|\frac{d{ \overline{O}}}{r^{\frac{3}{2}}}\|^2_{L^\infty}\big(
\|{ \tilde{u}}\|^2_{H^{2,\alpha}}+\|{ \tilde{O}-I}\|^2_{H^{2,\alpha+\frac{3}{2}}}\big)\\
\notag&+\|\frac{\partial^{(J)}\partial(dO)}{r^{\alpha+1}}\|^2_{L^2}
+\|\frac{\partial(dO)}{r^{\alpha+\frac{5}{2}}}\|^2_{L^2}
+\mathcal{E}_0\|\frac{\partial(dO)}{r^{\alpha+1}}\|^2_{L^2}\\
\notag&+\big\|\frac{\partial^{(J)}\big[({ \overline{O}+\tilde{\overline{O}}})d{ \overline{O}}\partial^2(\tilde{u}_\nu)_{ij}\big]}{r^{\alpha-2}}\big\|^2_{L^2}+\|\frac{\partial^{(J)}G(d{ \overline{u}},d{ \overline{O}})}{r^{\alpha-2}}\|^2_{L^2}\\
\notag\lesssim&\;E_{2,\alpha+1}[d{ u}]
+\mathcal{E}_0E_{2,\alpha}[d{ u}]
+\|dO\|^2_{H^{2,\alpha+\frac{5}{2}}}
+\mathcal{E}_0\|dO\|^2_{H^{2,\alpha+\frac{3}{2}}}\\
&\notag+\mathcal{E}_0\|d{ \overline{O}}\|^2_{H^{2,\alpha+\frac{3}{2}}}
+\big\|\frac{\partial^{(J)}\big[({ \overline{O}+\tilde{\overline{O}}})d{ \overline{O}}\partial^2(\tilde{u}_\nu)_{ij}\big]}{r^{\alpha-2}}\big\|^2_{L^2}+\|\frac{\partial^{(J)}G(d{ \overline{u}},d{ \overline{O}})}{r^{\alpha-2}}\|^2_{L^2}
\end{align}
We proceed to the problematic term $(\overline{O}+\tilde{\overline{O}})d{ \overline{O}}\partial^2(\tilde{u}_\nu)_{ij}$ which can be controlled only in $H^1$:
\begin{align}\label{mainest3}
&\big\|\frac{\partial^{(J)}\big[({ \overline{O}+\tilde{\overline{O}}})d{ \overline{O}}\partial^2(\tilde{u}_\nu)_{ij}\big]}{r^{\alpha-2}}\big\|^2_{L^2}\\
\lesssim\notag&\;\|\partial^{(J)}{(\overline{O}+\tilde{\overline{O}})}\|^2_{L^\infty}
\|d{ \overline{O}}\|^2_{L^\infty}\|\frac{\partial^2(\tilde{u}_\nu)_{ij}}{r^{\alpha-2}}\|^2_{L^2}
+\|{ \overline{O}+\tilde{\overline{O}}}\|^2_{L^\infty}\|\frac{\partial^{(J)}(d{ \overline{O}})}{r^{\frac{\alpha}{2}-1}}\|^2_{L^4}\\
\notag&\cdot\|\frac{\partial^2(\tilde{u}_\nu)_{ij}}{r^{\frac{\alpha}{2}-1}}\|^2_{L^4}
+\|{ \overline{O}+\tilde{\overline{O}}}\|^2_{L^\infty}
\|\frac{d{ \overline{O}}}{r}\|^2_{L^\infty}\|\frac{\partial^{(J)}\partial^2(\tilde{u}_\nu)_{ij}}{r^{\alpha-2}}\|^2_{L^2}\\
\lesssim&\;\tag{employing the $L^4$ estimate (\ref{L4})} (\mathcal{E}^2_0+\alpha^3\mathcal{E}^2_0)\|d{ \overline{O}}\|^2_{H^{2,\alpha+\frac{3}{2}}}
\end{align}
Finally, plugging in the nonlinearity (\ref{G}), we have the bound
\begin{align}\label{mainest4}
&\big\|\frac{\partial^{(J)}G(d{ \overline{u}},d{ \overline{O}})}{r^{\alpha-2}}\big\|^2_{L^2}\\
\notag\lesssim&\;\big\|\frac{\partial^{(J)}\big(\Gamma_\frac{3}{2}d{ \overline{u}}({ \overline{u}}+{ \tilde{\overline{u}}}+{ \overline{u}}^2+{\tilde{\overline{u}}}^2+{ \overline{u}\tilde{\overline{u}}})\big)}{r^{\alpha-2}}\big\|^2_{L^2}
+\big\|\frac{\partial^{(J)}\big({ \overline{O}}{ \overline{u}}\partial(d{ \overline{u}})\big)}{r^{\alpha-2}}\big\|^2_{L^2}\\
\notag&+\big\|\frac{\partial^{(J)}\big({ \overline{O}}d{ \overline{u}}\partial{ \tilde{\overline{u}}}
+d{ \overline{O}}{ \tilde{\overline{u}}}\partial{ \tilde{\overline{u}}}
+{ \overline{O}}\partial({ \overline{O}})\partial(d{ \overline{u}})\big)}{r^{\alpha-2}}\big\|^2_{L^2}\\
&\notag+\big\|\frac{\partial^{(J)}\big({ \overline{O}}\partial(d{ \overline{O}})\partial{ \tilde{\overline{u}}}
+d{ \overline{O}}\partial({ \tilde{\overline{O}}})\partial{ \tilde{\overline{u}}}\big)}{r^{\alpha-2}}\big\|^2_{L^2}\\
\lesssim&\;\notag(\mathcal{E}_0+\mathcal{E}_0^2)E_{2,\alpha}[d{ \overline{u}}]+(\mathcal{E}_0+1)\|d{ \overline{u}}\|^2_{L^\infty}\|\frac{\partial^{(J)}{ \partial\tilde{\overline{u}}}}{r^{\alpha-2}}\|^2_{L^2}+\|\frac{\partial^2{ (\overline{O})}\partial(d{ \overline{u}})}{r^{\alpha-2}}\|^2_{L^2}\\
\notag&+(\mathcal{E}_0+\mathcal{E}_0^2)\|d{ \overline{O}}\|^2_{H^{2,\alpha+\frac{3}{2}}}+\|\frac{\partial(d{ \overline{O}})\partial^2{ \tilde{\overline{u}}}}{r^{\alpha-2}}\|^2_{L^2}
+\|d{ \overline{O}}\|^2_{L^\infty}\mathcal{E}_0^2\\
\lesssim&\;\notag (\mathcal{E}_0+\mathcal{E}_0^2)\big(E_{2,\alpha}[d{ \overline{u}}]+\|d{ \overline{O}}\|^2_{H^{2,\alpha+\frac{3}{2}}}\big)+\|\frac{\partial^2{ \overline{O}}}{r^{\frac{\alpha}{2}-1}}\|^2_{L^4}\|\frac{\partial(d{ \overline{u}})}{r^{\frac{\alpha}{2}-1}}\|^2_{L^4}\\
&\notag+\|\frac{\partial(d{ \overline{O}})}{r^{\frac{\alpha}{2}-1}}\|^2_{L^4}\|\frac{\partial^2{ \tilde{\overline{u}}}}{r^{\frac{\alpha}{2}-1}}\|^2_{L^4}\\
\lesssim&\;\tag{by the $L^4$ estimate (\ref{L4})} (\mathcal{E}_0+\mathcal{E}_0^2+\alpha^3\mathcal{E}_0)\big(E_{2,\alpha}[d{ \overline{u}}]+\|d{ \overline{O}}\|^2_{H^{2,\alpha+\frac{3}{2}}}\big)
\end{align}
\end{proof}

\section{The constraint equations in a singular background of unbounded mean curvature}\label{singconst}

In this section we prove Theorem \ref{thmB}, our main stability result for the constraint equations (\ref{const}). The proof is based on mapping properties of the constraint map \cite{Mon1} \cite{CorSch} \cite{ChDel} and the implicit function theorem. Although similar results have been achieved in the smooth case and some rough backgrounds (see \cite{ChDel} for a general exposition), to our knowledge, the singular Schwarzschild background (\S\ref{Schwarz}) eludes the standard references in the literature.

It would be interesting to obtain perturbations $(\overline{g},K)$ of the Schwarzschild initial data set $(^S\overline{g},{^SK})$ on the singular hypersurface $\Sigma_0$, which satisfy the assumptions of Theorem \ref{thmA}, in a more constructive way. A general such construction would make use of the conformal method \cite{Choq2} that we do not employ here. However, one of the main obstructions to overcome in this approach is that the mean curvature $\text{tr}_{\overline{g}}{K}$ of the perturbation is unbounded. In fact, one can check (\S\ref{Schwarz}) that 
\begin{align*}
\text{tr}_{\overline{g}}{K}\not\in L^p(\Sigma_0),&&p\ge\frac{5}{3}.
\end{align*}
The results in the literature of the constraints using the conformal method were restricted in the past to the constant mean curvature (CMC) or `near CMC' regime \cite{Choq2}. Recently, there has been a number of advances to the case of large mean curvature, `far from CMC', \cite{HNT,Max,DIMM}. Yet these results contain certain regularity assumptions which in particular imply that the mean curvature of the initial data set is in $L^\infty$ and therefore do not apply to our case.

We consider below various tensors living in the weighted spaces $H^{s,\alpha}(\Sigma_0)$ we have introduced (\ref{Hsa}). It will always be implied that this is understood by assuming their components, evaluated w.r.t. the Schwarzschild frame (\ref{Sframe}), lie in the same spaces. 

\subsection{The constraint map; Linearization and stability}\label{linstab}

Let 
\begin{align}\label{Psi}
\Psi&:H^{4,\alpha+\frac{3}{2}}(\Sigma_0)\times H^{3,\alpha}(\Sigma_0)\to H^{2,\alpha-\frac{3}{2}}(\Sigma_0)\times H^{2,\alpha-\frac{3}{2}}(\Sigma_0),\\
\notag&\Psi(d g_{ij},d K_{ij}):=\big({\overline{\text{R}}}-|{K}|^2+(\text{tr}_{{\overline{g}}}{K})^2,\;{\overline{{\nabla}}}^j{K}_{ij}-{\overline{\nabla}}_i\text{tr}_{{\overline{g}}}{K}\big)
\end{align}
denote the constraint map of the perturbed initial data
\begin{align}\label{deltagK}
dg_{ij}:={\overline{g}}_{ij}-{^S\overline{g}_{ij}}&&dK_{ij}:={K}_{ij}-{^SK_{ij}},&&i,j=1,2,3,
\end{align}
on $\Sigma_0:=(-\epsilon,\epsilon)_x\times r(\tau=0)\mathbb{S}^2_{\theta,\phi}$.
\begin{lemma}\label{BdC1}
$\Psi$ is well-defined, bounded and $C^1$ (Fr\'echet).
\end{lemma}
\begin{proof}
First, we express $\Psi$ in terms of differences of perturbed and Schwarzschild components. Since the Schwarzschild pair $(^S\overline{g},{^SK}$) satisfies the constraints (\ref{const}), we write schematically:
\begin{align}\label{Psidiff}
\Psi(d g, d K)=&\;\big( \partial(d { A})+ (d \partial)({ ^SA})+{^S A}d{ A}+({d A})^2+{^S K}d{ K}+({d K})^2,\\
\notag &\;\;\; \partial(d { K})+ (d \partial)({ ^SK})+{^S K}d { A}
+{^SA}{ d K}+{ d A}{ d K}\big),
\end{align}
where $d{ A}:={ {A}}-{^SA}$. The boundedness of $\Psi$ now follows in part from Lemma \ref{Sobprop}, using the asymptotics
(\ref{overasym1}),(\ref{overasym2}), and by applying the $L^\infty$ bound (\ref{Linfty}) on the quadratic terms $({d K})^2$, $({d A})^2$, ${d A}{d K}$. Realizing that ${ d A}$ is $C^1$ in the variable $d g$ (by expressing ${ d A}$ in terms of ${d g}$), we conclude that $\Psi$ is also $C^1$ in the above weighted spaces (taking into account the previous asymptotics and properties used).
\end{proof}
\begin{remark}\label{Hsa(x)}
In view of the asymptotics (\ref{rasym}) of the radius function $r$, at $\tau=0$, we may replace the $H^{s,\alpha}$ norm (\ref{Hsa}), at our convenience, with the equivalent norm
\begin{align}\label{Hsaequiv}
\|v\|^2_{H^{s,\alpha}(\Sigma_0)}\overset{\text{equiv}}\approx\sum_{|J|=k\leq s}
\int_{\Sigma_0}\frac{(^S\overline{\nabla}^{(J)}v)^2}{x^{2\alpha-3(k-1)}}d\mu_{^S\overline{g}},
\end{align}
where $^S\overline{\nabla}$ is the induced Schwarzschild connection on $(\Sigma_0,^S\overline{g})$; we replace $\partial^{(J)}$ with $^S\overline{\nabla}^{(J)}$ in (\ref{Hsa}) as well, since the weights in the norm can exactly tolerate the most singular coefficients (\ref{overasym1}) of any additional generated terms. 
\end{remark}
Our strategy is to fix the boundary components (close to Schwarzschild) and then solve for the variables in the interior via a perturbation. For this purpose 
we introduce the functions 
\begin{align}\label{split}
d g^{(0)}\in H^{4,\alpha+\frac{3}{2}}\cap H^1_0,\;\;d K^{(0)}\in H^{3,\alpha}\cap H^1_0,\;\;d g^{(1)}\in H^{4,\alpha+\frac{3}{2}},\;\;d K^{(1)}\in H^{3,\alpha}\\
\notag d g=d g^{(0)}+d g^{(1)},\qquad d K=d K^{(0)}+d K^{(1)},
\end{align}
where by $H^1_0$ we denote the set of functions in $H^1(\Sigma_0)$ having zero trace on the boundary of $\Sigma_0$.\footnote{The boundary $\partial\Sigma_0$ has two components diffeomorphic to the sphere.}
We fix $d g^{(1)},d K^{(1)}$ and study $\Psi$ under variations of the variables $d g^{(0)},d K^{(0)}$.  Whence, we can view $\Psi$ as a map
\begin{align}\label{Psi2}
\notag\Psi:\big((H^{4,\alpha+\frac{3}{2}}&\cap H^1_0)\times (H^{3,\alpha}\cap H^1_0)\big)\times(H^{4,\alpha+\frac{3}{2}}\times H^{3,\alpha})\to H^{2,\alpha-\frac{3}{2}}\times H^{2,\alpha-\frac{3}{2}},\\
&\Psi(d g^{(0)},d K^{(0)},d g^{(1)},d K^{(1)}):=\Psi(d g,d K).
\end{align}

Let
\begin{align}\label{DPsi}
D\Psi:=D_{(d g^{(0)},d K^{(0)})}\Psi_{(0,0,0,0)}:H^{4,\alpha+\frac{3}{2}}\times H^{3,\alpha}\to H^{2,\alpha-\frac{3}{2}}\times H^{2,\alpha-\frac{3}{2}}
\end{align}
be the Fr\'echet derivative of $\Psi$ w.r.t. the variables $(d g^{(0)},d K^{(0)})$ about zero, that is, the Schwarzschild initial data; $d g^{(0)}=d K^{(0)}=d g^{(1)}=d K^{(1)}=0$. The following theorem is our stability result for the constraint equations.
\begin{theorem}\label{conststab}
Let $\alpha>0$ be sufficiently large, consistent with Theorems \ref{mainthm}, \ref{mainthm2}. Then $D\Psi$ is surjective, and the level set 
\begin{align}\label{levset}
\bigg(\Psi\bigg|_{(\cdot,\cdot,d g^{(1)},d K^{(1)})}\bigg)^{-1}(\{0\})
\end{align}
is non-empty, for every pair $(d g^{(1)},d K^{(1)})$ in a sufficiently small ball of $H^{4,\alpha+\frac{3}{2}}\times H^{3,\alpha}$.
\end{theorem}
\begin{remark}\label{constexist}
The preceding theorem implies the existence of non-spherically symmetric perturbations $({\overline{g}},{K}_{ij})$\footnote{This is easily seen by considering non-spherically symmetric $(d g^{(1)},d K^{(1)})$.} of the Schwarzschild initial data set $(^S\overline{g},{^SK}_{ij})$ on $\Sigma_0$, compatible with the constraint equations (\ref{const}), which lie inside the weighted spaces $H^{4,\alpha+\frac{3}{2}}(\Sigma_0)\times H^{3,\alpha}(\Sigma_0)$, satisfying the assumptions of the local-in time-existence Theorems \ref{mainthm}, \ref{mainthm2}.
\end{remark}
\begin{remark}\label{halfS}
We can construct perturbed initial data sets $({\overline{g}},{K})$ verifying the constraints (\ref{const}) and the assumptions of Theorems \ref{mainthm}, \ref{mainthm2}, which are identical to the Schwarzschild initial data set $(^S\overline{g},{^SK})$ on $\Sigma_0\cap\{x<0\}$ and not spherically symmetric on the other `half' $\Sigma_0\cap\{x>0\}$ (and vice versa) by applying the arguments in this section only to the $\{x>0\}$ `piece' of the initial data.
\end{remark}
Unfortunately, the formula of $D\Psi$ is quite complex \cite{Mon1} \cite{CorSch} and it is not easy to prove directly the surjectivity of $D\Psi$. We do not present the formula of $D\Psi$ here, since we are not going to use it. Instead, we follow an indirect argument of \cite{CorSch}, which we adapt in our context. We consider the restriction of $\Psi$ (\ref{Psi}) in the subdomain 
\begin{align}\label{confdata}
{\overline{g}}=\varphi^4{^S\overline{g}}&&{K}_{ij}=\varphi^2({^SK}+\mathcal{L}_X{^S\overline{g}})_{ij},
\end{align}
where $(\mathcal{L}_X{^S\overline{g}})_{ij}:={^S\overline{\nabla}}_iX_j+{^S\overline{\nabla}}_jX_i-{^S\overline{g}}_{ij}{^S\overline{\nabla}}^kX_k$. Then the constraint map (\ref{Psi}) becomes (\cite{CorSch})
\begin{align}\label{Psihat}
\notag\hat{\Psi}:(H^{4,\alpha+\frac{3}{2}}\cap H^1_0)&\times(H^{4,\alpha+\frac{3}{2}}\cap H^1_0)\to H^{2,\alpha-\frac{3}{2}}\times H^{2,\alpha-\frac{3}{2}},\\
\hat{\Psi}(\varphi-1,X)=&\;
\big(\varphi^{-5}\big[-8\Delta_{{^S\overline{g}}}\varphi
+\varphi^S\overline{\text{R}}-\varphi\,\big|{^SK}
+\mathcal{L}_X{^S\overline{g}}\,\big|_{{^S\overline{g}}}^2\\
\notag&+\frac{1}{2}\varphi(\text{tr}_{{^S\overline{g}}}{^SK}+\mathcal{L}_X{^S\overline{g}})^2\big],\\
\notag&\varphi^{-2}[{^S\overline{\nabla}}^j({^SK}+\mathcal{L}_X{^S\overline{g}})_{ij}+4\varphi^{-1}{(^SK+\mathcal{L}_X{^S\overline{g}})_i}^j{^S\overline{\nabla}}_j\varphi\\
\notag&-2\varphi^{-1}{^S\overline{\nabla}}_i\varphi(\text{tr}_{{^S\overline{g}}}{^SK}+\mathcal{L}_X{^S\overline{g}})
]\big)
\end{align}
Further, the linearization of $\hat{\Psi}$ at $(0,0)$ reads
\begin{align}\label{DPsihat}
\notag D\hat{\Psi}(\eta,Y)=\big(&-8\Delta_{{^S\overline{g}}}\eta+2(\text{tr}_{{^S\overline{g}}}{^SK})^2\eta-2K^{ij}(\mathcal{L}_Y{^S\overline{g}})_{ij}+\text{tr}_{{^S\overline{g}}}{^SK}({^S\overline{\nabla}}^iY)_i,\\
&\;{^S\overline{\nabla}}^j(\mathcal{L}_Y{^S\overline{g}})_{ij}+4{{^SK}_i}^j{^S\overline{\nabla}}_j\eta-2\text{tr}_{{^S\overline{g}}}{^SK}{^S\overline{\nabla}}_i\eta-2{^S\overline{\nabla}}^j{^SK}_{ij}\eta\big)
\end{align}
Recall (\ref{overasym1}), (\ref{rasym}) to find the leading asymptotics, as $x\rightarrow0^+$, $\tau=0$, of the singular Schwarzschild components of $D\hat{\Psi}(\eta,Y)$:
\begin{align}\label{constcoeff}
|{^SK}_{ij}|,\;|\text{tr}_{{^S\overline{g}}}{^SK}|\lesssim\frac{1}{x^\frac{3}{2}}&&|^S\overline{\nabla}^{(J)}{^SK}_{ij}|\lesssim\frac{1}{x^{\frac{3}{2}(|J|+1)}}
\end{align}

A key ingredient in the proof of Theorem \ref{conststab} is the following proposition.
\begin{proposition}\label{Fred}
The operator 
\begin{align}\label{DPsihatsp}
D\hat{\Psi}:H^{4,\alpha+\frac{3}{2}}\cap H^1_0\to H^{2,\alpha-\frac{3}{2}}
\end{align}
is Fredholm. In particular, it has finite dim kernel and cokernel.
\end{proposition}
Assuming Proposition \ref{Fred}, we proceed to the proof of Theorem \ref{conststab}.
\begin{proof}[Proof of Theorem \ref{conststab}]
Evidently, $\text{range}D\Psi\supset\text{range} D\hat{\Psi}$ and hence by Proposition \ref{Fred} $D\Psi$ has finite dimensional cokernel. It follows that range$D\Psi$ is closed. 

Since range$D\Psi$ is closed, range$D\Psi=(\text{kernel}D\Psi^*)^\perp$; $D\Psi^*$ being the adjoint of $D\Psi$. Therefore, in order to prove surjectivity, it suffices to show that $\text{kernel}D\Psi^*=\{0\}$. We now recall a lemma from \cite{Mon1}:
\begin{lemma}\label{Monc}
The kernel of $D\Psi^*$ is in one-to-one correspondence with the set of Killing vector fields of the Einsteinian vacuum development of $(\overline{g},K)$, i.e., the Schwarzschild region (\ref{fol}) foliated by $\{\Sigma_\tau\}_{\tau\in[0,T]}$.
\end{lemma}
Picking $\alpha$ large enough, we can guarantee that none of the Schwarzschild Killing vector fields ($\partial_t,\partial_\phi,\sin\phi\,\partial_\theta+\cot\theta\cos\phi\,\partial_\phi,\cos\phi\,\partial_\theta-\cot\theta\sin\phi\,\partial_\phi$) lie in $H^{2,\alpha-\frac{3}{2}}(\Sigma_0)$, the domain of $D\Psi^*$. Therefore, for this $\alpha$ the condition kernel$D\Psi^*=0$ is verified and hence 
$D\Psi$ is surjective. The second part of Theorem \ref{conststab} follows by the implicit function theorem.
\end{proof}

\subsection{Proof of Proposition \ref{Fred}}\label{confFred}

We realize the following plan (which is an adaptation of a usual argument from the non-singular case): Given $\sigma>0$ we define the bounded operator (see Lemma \ref{BdC1})
\begin{align}\label{S}
S:=D\hat{\Psi}+\frac{\sigma}{x^3}(I_\eta,-I_Y):H^{4,\alpha+\frac{3}{2}}\cap H^1_0\to H^{2,\alpha-\frac{3}{2}},
\end{align}
where $I_\eta,I_Y$ are the identity maps for the variables $\eta,Y$ respectively. Then we show:
\begin{proposition}\label{Siso}
$S$ is an isomorphism for $\sigma>0$ sufficiently large.
\end{proposition}
Using the preceding proposition, we obtain the Fredholm property claimed in Proposition \ref{Fred} as follows:\\
Since $S$ is an isomorphism (for large $\sigma>0$), it has a bounded inverse. In fact,
\begin{align}\label{S-1}
S^{-1}:H^{2,\alpha-\frac{3}{2}}\to H^{4,\alpha+\frac{3}{2}}\cap H^1_0\overset{\text{compact}}\hookrightarrow H^{2,\alpha-\frac{3}{2}}
\end{align}
is a compact operator by a weighted version of Rellich's theorem. By definition (\ref{S}) the operators
\begin{align}\label{DPsihatS}
D\hat{\Psi} \circ S^{-1}-I=&-\frac{\sigma}{x^3}(I_\eta,-I_Y)\circ S^{-1}: H^{2,\alpha-\frac{3}{2}}\to  H^{2,\alpha-\frac{3}{2}}\\
\notag S^{-1}\circ D\hat{\Psi}-I=&-S^{-1}\circ\frac{\sigma}{x^3}(I_\eta,-I_Y):H^{4,\alpha+\frac{3}{2}}\cap H^1_0\to H^{4,\alpha+\frac{3}{2}}\cap H^1_0
\end{align}
are also compact, since the RHSs of (\ref{DPsihatS}) consist of compositions of bounded with compact operators. Thus, $D\hat{\Psi}$ is Fredholm.\\
A note is in order here. The key ingredient in the previous proof is the invertibility of the operator $S$ (\ref{S}). As we shall see below we are able to prove Proposition \ref{Siso} thanks to the precise analogy in the leading blow up orders (\ref{constcoeff}) of the Schwarzschild components appearing in the expression of the linearized map $D\hat{\Psi}$ (\ref{DPsihat}).

We divide the proof of Proposition \ref{Siso} in two steps. Recall the formula (\ref{DPsihat}).
\begin{proposition}\label{1-1}
For every $h,z_i\in H^{2,\alpha-\frac{3}{2}}$, $i=1,2,3$, the system
\begin{align}\label{Spde}
-8\Delta_{{^S\overline{g}}}\eta+2(\text{tr}_{{^S\overline{g}}}{^SK})^2\eta-2{^SK}^{ij}(\mathcal{L}_Y{^S\overline{g}})_{ij}+\text{tr}_{{^S\overline{g}}}{^SK}({^S\overline{\nabla}}^iY)_i+\frac{\sigma}{x^3}\eta=h\\
\notag{^S\overline{\nabla}}^j(\mathcal{L}_Y{^S\overline{g}})_{ij}+4{{^SK}_i}^j{^S\overline{\nabla}}_j\eta-2\text{tr}_{{^S\overline{g}}}{^SK}{^S\overline{\nabla}}_i\eta-2{^S\overline{\nabla}}^j{^SK}_{ij}\eta-\frac{\sigma}{x^3}Y_i=z_i
\end{align}
has a unique weak solution $\eta,Y\in H^{1,\alpha+\frac{3}{2}}_0$, provided $\sigma>0$ is sufficiently large.
\end{proposition}
\begin{proof}
Define the bilinear form
\begin{align}\label{B[eta,Y]}
\mathcal{B}:&\;H^{1,\alpha+\frac{3}{2}}_0\times H^{1,\alpha+\frac{3}{2}}_0\to\mathbb{R}\\
\notag&\mathcal{B}[(\eta,Y),(\xi,X)]=\big(\text{LHS of (\ref{Spde})},(\frac{\xi}{x^{2\alpha+3}},-\frac{X}{x^{2\alpha+3}})\big)_{L^2}
\end{align}
We want to apply Lax-Milgram. For that we must establish the two estimates:
\begin{align}\label{Lmest1}
\big|\mathcal{B}[(\eta,Y),(\xi,X)]\big|\lesssim \big(\|\eta\|_{H^{1,\alpha+\frac{3}{2}}}+\|Y\|^2_{H^{1,\alpha+\frac{3}{2}}}\big)\big(\|\xi\|_{H^{1,\alpha+\frac{3}{2}}}+\|X\|^2_{H^{1,\alpha+\frac{3}{2}}}\big)
\end{align}
and
\begin{align}\label{LMest2}
\mathcal{B}[(\eta,Y),(\eta,Y)]\gtrsim \|\eta\|^2_{H^{1,\alpha+\frac{3}{2}}}+\sum_i\|Y_i\|^2_{H^{1,\alpha+\frac{3}{2}}}
\end{align}
The bound (\ref{Lmest1}) follows easily by Cauchy-Schwartz, taking into account the asymptotics (\ref{constcoeff}) of the coefficients of (\ref{Spde}). Hence, it suffices to show (\ref{LMest2}). We proceed by integrating by parts, 
employing the asymptotics (\ref{constcoeff}):
\begin{align*}
\mathcal{B}[(\eta,Y),(\eta,Y)]=\big(\text{LHS of (\ref{Spde})},(\frac{\eta}{x^{2\alpha+3}},-\frac{Y}{x^{2\alpha+3}})\big)_{L^2}\ge\\
8\int_{\Sigma_0}\frac{|{^S\overline{\nabla}} \eta|^2}{x^{2\alpha+3}}d\mu_{^S\overline{g}}-C\alpha\int_{\Sigma_0}\frac{|\eta{^S\overline{\nabla}}_j\eta|}{x^{2\alpha+3+\frac{1}{2}}}d\mu_{^S\overline{g}}+\sigma\int_{\Sigma_0}\frac{\eta^2}{x^{2\alpha+6}}d\mu_{^S\overline{g}}\\
-\;C\int_{\Sigma_0}\frac{|\eta{^S\overline{\nabla}}_jY_k|}{x^{2\alpha+3+\frac{3}{2}}}d\mu_{^S\overline{g}}
+\int_{\Sigma_0}\frac{|{^S\overline{\nabla}} Y|^2}{x^{2\alpha+3}}d\mu_{^S\overline{g}}-\int_{\Sigma_0}|\overline{R}_{ki}|\frac{|Y^kY^i|}{x^{2\alpha+3}}d\mu_{^S\overline{g}}\\
-\;C\alpha\int_{\Sigma_0}\frac{|Y_k{^S\overline{\nabla}}_aY_b|}{x^{2\alpha+3+\frac{1}{2}}}d\mu_{^S\overline{g}}
+\sigma\int_{\Sigma_0}\frac{Y^2}{x^{2\alpha+6}}d\mu_{^S\overline{g}}\\
-\;C\int_{\Sigma_0}\frac{|Y_i{^S\overline{\nabla}}_j\eta|}{x^{2\alpha+3+\frac{3}{2}}}d\mu_{^S\overline{g}}- C\int_{\Sigma_0}\frac{|\eta Y_i|}{x^{2\alpha+5}}d\mu_{^S\overline{g}}\ge\\
\tag{utilizing the estimate $|^S\overline{\text{R}}_{ki}|\lesssim x^{-2}$ and applying C-S}
\int_{\Sigma_0}\frac{|{^S\overline{\nabla}} \eta|^2}{x^{2\alpha+3}}d\mu_{^S\overline{g}}
+\frac{1}{2}\int_{\Sigma_0}\frac{|{^S\overline{\nabla}}Y|^2}{x^{2\alpha+3}}d\mu_{^S\overline{g}}
+(\sigma-C\alpha^2)\bigg[\int_{\Sigma_0}\frac{\eta^2}{x^{2\alpha+6}}d\mu_{^S\overline{g}}\\
+\int_{\Sigma_0}\frac{Y^2}{x^{2\alpha+6}}d\mu_{^S\overline{g}}\bigg]
\end{align*}
Taking $\sigma$ large enough, $\sigma\gtrsim\alpha^2$, we arrive at (\ref{LMest2}). Thus, the Lax-Milgram theorem can be applied for the system (\ref{Spde}), from which the conclusion of the present proposition follows.
\end{proof}
\begin{remark}
Away from $x=0$ the coefficients of (\ref{Spde}) are smooth and bounded (locally). By standard interior elliptic regularity theory we derive that the weak solution in Proposition \ref{1-1} $\eta,Y$ lies in 
$H^4_{loc}(\Sigma_0\smallsetminus \{x=0\})$. 
\end{remark}
Next, we show that the solution of (\ref{Spde}) is in fact $\eta,Y\in H^{4,\alpha+\frac{3}{2}}$. This implies that $S$ (\ref{S})  is one to one and onto, proving Proposition \ref{Siso}.
\begin{proposition}\label{ellreg}
The weak solution $\eta,Y\in H^{1,\alpha+\frac{3}{2}}_0$ of (\ref{Spde}) furnished by Proposition \ref{1-1} satisfies $\eta,Y\in H^{4,\alpha+\frac{3}{2}}$.
\end{proposition}
\begin{proof}
$\eta,Y\in H^{2,\alpha+\frac{3}{2}}$: Multiplying (\ref{Spde}) with $(\frac{{^S\overline{\nabla}}_{22}\eta}{x^{2\alpha}},\frac{{^S\overline{\nabla}}_{22}Y_i}{x^{2\alpha}})$, 
commuting covariant derivatives and integrating by parts twice we arrive at
\begin{align*}
\int_{\Sigma_0}\frac{|{^S\overline{\nabla}} {^S\overline{\nabla}}_2\eta|^2}{x^{2\alpha}}d\mu_{^S\overline{g}}-C\alpha\int_{\Sigma_0}\frac{|\nabla_{a2}\eta{^S\overline{\nabla}}_2\eta|}{x^{2\alpha+\frac{1}{2}}}d\mu_{^S\overline{g}}+\sigma\int_{\Sigma_0}\frac{|{^S\overline{\nabla}}_2\eta|^2}{x^{2\alpha+3}}d\mu_{^S\overline{g}}\\
+\int_{\Sigma_0}{ ^S\overline{\text{Ric}}}\frac{{^S\overline{\nabla}}_a\eta{^S\overline{\nabla}}_2\eta}{x^{2\alpha}}d\mu_{^S\overline{g}}-C\int_{\Sigma_0}\frac{|{^S\overline{\nabla}}_{22}\eta{^S\overline{\nabla}}_jY_k|}{x^{2\alpha+\frac{3}{2}}}\leq \int_{\Sigma_0}\frac{|h{^S\overline{\nabla}}_{22}\eta|}{x^{2\alpha}}d\mu_{^S\overline{g}}\,, \\
\int_{\Sigma_0}\frac{|{^S\overline{\nabla}}_j{^S\overline{\nabla}}_2Y_i|^2}{x^{2\alpha}}d\mu_{^S\overline{g}}
-C\alpha\int_{\Sigma_0}\frac{|{^S\overline{\nabla}}_{a2}Y_b{^S\overline{\nabla}}_2Y_i|}{x^{2\alpha+\frac{1}{2}}}d\mu_{^S\overline{g}}
+\sigma\int_{\Sigma_0}\frac{|{^S\overline{\nabla}}_2Y_i|^2}{x^{2\alpha+3}}d\mu_{^S\overline{g}}\\
+\int_{\Sigma_0}{^S\overline{\text{Riem}}}\frac{{^S\overline{\nabla}}_aY_j{^S\overline{\nabla}}_2Y_i}{x^{2\alpha}}d\mu_{^S\overline{g}}+\int_{\Sigma_0}
{^S\overline{\text{Riem}}}\frac{{^S\overline{\nabla}}_{2b}Y_jY_k}{x^{2\alpha}}d\mu_{^S\overline{g}}\\
-\; C\int_{\Sigma_0}\frac{|{^S\overline{\nabla}}_a\eta{^S\overline{\nabla}}_{22}Y_i|}{x^{2\alpha+\frac{3}{2}}}d\mu_{^S\overline{g}}
-C\int_{\Sigma_0}\frac{|\eta{^S\overline{\nabla}}_{22}Y_i|}{x^{2\alpha+2}}d\mu_{^S\overline{g}}
\leq\int_{\Sigma_0}\frac{|z_i{^S\overline{\nabla}}_{22}Y_i|}{x^{2\alpha}}d\mu_{^S\overline{g}}
\end{align*}
Using the bound ${^S\overline{\text{Riem}}}\lesssim\frac{1}{x^2}$ and applying Cauchy's inequality we deduce for $\sigma$ large ($\sigma\gtrsim\alpha^2$):
\begin{align*}
\int_{\Sigma_0}\frac{|{^S\overline{\nabla}} {^S\overline{\nabla}}_2\eta|^2}{x^{2\alpha}}d\mu_{^S\overline{g}}
+\sum_i\int_{\Sigma_0}\frac{|{^S\overline{\nabla}}{^S\overline{\nabla}}_2Y_i|^2}{x^{2\alpha}}d\mu_{^S\overline{g}}\\
\lesssim\|\eta\|^2_{H^{1,\alpha+\frac{3}{2}}}+\|Y\|^2_{H^{1,\alpha+\frac{3}{2}}}
+\int_{\Sigma_0}\frac{h^2}{x^{2\alpha}}d\mu_{^S\overline{g}}
+\sum_i\int_{\Sigma_0}\frac{z_i^2}{x^{2\alpha}}d\mu_{^S\overline{g}}<+\infty,
\end{align*}
since $h,z_i\in H^{2,\alpha-\frac{3}{2}}$. Likewise for the $\partial_3$ (i.e., the other rotational direction). 
The $^S\overline{\nabla}_{11}$ derivative of the functions $\eta,Y_i$ is estimated separately: From (\ref{Spde}) we have 
\begin{align*}
\sum_i\int_{\Sigma_0}\frac{|{^S\overline{\nabla}}_{11}Y_i|^2}{x^{2\alpha}}d\mu_{^S\overline{g}}+\int_{\Sigma_0}\frac{|{^S\overline{\nabla}}_{11}\eta|^2}{x^{2\alpha}}d\mu_{^S\overline{g}}\\
\lesssim \sum_{a=2,3}\bigg[\int_{\Sigma_0}\frac{|{^S\overline{\nabla}}{^S\overline{\nabla}}_aY_b|^2}{x^{2\alpha}}d\mu_{^S\overline{g}}+\int_{\Sigma_0}\frac{|{^S\overline{\nabla}}_{aa}\eta|^2}{x^{2\alpha}}d\mu_{^S\overline{g}}\bigg]
+\sigma\|Y\|^2_{H^{1,\alpha+\frac{3}{2}}}\\
+\;\sigma\|\eta\|^2_{H^{1,\alpha+\frac{3}{2}}}
+\|\frac{h}{x^\alpha}\|^2_{L^2}+\sum_i\|\frac{z_i}{x^\alpha}\|^2_{L^2}<+\infty.
\end{align*}
Thus, $\eta,Y_i\in H^{2,\alpha+\frac{3}{2}}$. The full $H^{4,\alpha+\frac{3}{2}}$ estimate is obtained by differentiating the system (\ref{Spde}) in the rotational directions $\partial_2,\partial_3$, applying the same argument with the analogous weights and then/ differentiating (\ref{Spde}) in the $\partial_1$ direction.
\end{proof}

\appendix

\section{Changing frames freedom; Propagating identities; Retrieving the EVE from the reduced equations}\label{app1}

Given a spacetime $(\mathcal{M}^{1+3},g)$ and an orthonormal frame $\{e_i\}^3_0$, one may change to a Lorentz gauge frame $\{\tilde{e}_i\}^3_0$ by solving the following semilinear system of equations, which is derived by taking the divergence of (\ref{X(O)}):
\begin{align*}
\square_g(O^l_a)=&\;(\text{div}{\tilde{A}_X)_a}^dO^l_d+\tilde{{A}}{ \partial}{ O}
+{ A}{ \partial}{ O}+O^k_a\text{div}{(A_X)_k}^l\\
\tag{by (\ref{gauge}) for $\tilde{{ A}}$}
=&\;\tilde{{ A}}^2{ O}+\tilde{{A}}{ \partial}{ O}
+{ A}{ \partial}{ O}+{ O}\text{div}{ A}\\
\tag{from (\ref{trans})}
=&\;{ O}^5{ A}^2+{ O}^3({ \partial}{ O})^2+{ A}{ O}^4{ \partial}{ O}
+{ A}{ O}^2{ \partial}{ O}\\
&+{ O}({ \partial}{ O})^2
+{ A}{ \partial}{ O}+{ O}\text{div}{ A_X},
\end{align*}
where the terms without indices in the RHS stand for an algebraic expression of  a finite number terms of the depicted type.
\begin{lemma}\label{gaugemap}
If the above system (which we write schematically as)
\begin{align}\label{boxO2}
\square_g(O^b_i)=&\;{ O}^5{ A}^2+{ O}^3({ \partial}{ O})^2+{ A}{ O}^4{ \partial}{ O}
+{ A}{ O}^2{ \partial}{ O}\\
\notag&+{ O}({ \partial}{ O})^2
+{ A}{ \partial}{ O}+{ O}\text{div}{ A}.
\end{align}
is well-posed in a certain solution space, then there exists a unique orthonormal frame 
\begin{align}\label{O2}
\tilde{e}_i=O^b_ie_b
\end{align}
with $O^b_i$ lying in that particular space, which is identical to $\{e_i\}^3_0$ on the initial hypersurface $\Sigma_0$, verifies the Lorentz gauge condition (\ref{gauge})
and such that the connection coefficients
$(\tilde{A}_0)_{ij}:=g(\nabla_{\tilde{e}_0}\tilde{e}_i,\tilde{e}_j)$, $i<j$,
are equal to a priori assigned functions on $\Sigma_0$; within the corresponding space of one order of regularity less than $O^b_i$.
\end{lemma}
\begin{proof}
It suffices to show that the initial data for (\ref{boxO2}) is uniquely determined by the assertions. We set
\begin{align}\label{initO}
O^b_i(\tau=0):={I_i}^b&&\text{(i.e., $\tilde{e}_i=e_i$ on $\Sigma_0$)}.
\end{align}
Let 
\begin{align}\label{initO2}
\tilde{e}_0(O^b_i)(\tau=0)=e_0(O^b_i)(\tau=0)=:h^b_i,&&h^b_im_{bj}=-h^b_jm_{bi}.
\end{align}
Then the transition formula (\ref{trans}) for $X=\tilde{e}_0$ reads
\begin{align}\label{transA_0}
(\tilde{A}_0)_{ij}(\tau=0)=(A_0)_{ij}(\tau=0)+h^b_im_{bj}.
\end{align}
Thus, the components
$(\tilde{A}_0)_{ij}$ can be freely prescribed initially by choosing $h^b_i$ in (\ref{initO2}) accordingly.
\end{proof}
\subsection{Proof of proposition \ref{IDprop}}\label{app1.1}

We will leave the reader to fill in the details for the fact that the solution $(A_\nu)_{ij},O^a_i$ of (\ref{boxA3}),(\ref{partial_0frame}) corresponds to a spacetime $(\mathcal{M}^{1+3},g)$. This is a consequence of the necessary initial assumption (\ref{initAO}). One such immediate consequence follows from (\ref{partial_0frame}) for $i=0$:
\begin{align}
\partial_0(O^a_0)=-O^b_0\Gamma^a_{[0b]},&&O^a_0(\tau=0)-{I_0}^a=0,
\end{align}
which implies $O^a_0={I_0}^a$ and hence $e_0=\partial_0$ everywhere, since $\Gamma^a_{[00]}=0$. The set of functions $O^a_i$ defines the orthonormal frame $\{e_i\}^3_0$ in $\mathcal{M}^{1+3}$ through (\ref{epartial}) and hence completely determines the metric $g$. What remains to be verified is that the connection coefficients of $\{e_i\}^3_0$ are indeed the $(A_\nu)_{ij}$'s of the given solution. In other words, we have to show that the connection $D$ induced by the solution set $(A_\nu)_{ij}$,
\begin{align}
D_{e_\nu}e_i:={(A_\nu)_i}^ke_k,
\end{align}
is the Levi-Civita connection $\nabla$ of the metric $g$. Formally, one cannot take this for granted. It has to be retrieved from the equations (\ref{boxA3}),(\ref{partial_0frame}) and the initial assumption (\ref{initAO}). For example, the compatibility of $D$ with respect to $g$ is encoded in the skew-symmetry of the $(A_\nu)_{ij}$'s
\begin{align}\label{Dg}
D(g)=0,&&\text{iff}&&(A_\nu)_{ij}+(A_\nu)_{ji}=0,
\end{align}
which also has to be verified, since it is a priori valid only initially (\ref{initAO}). The way to do this is by deriving the following new system of equations from (\ref{boxA3}) for the symmetric sums:
\begin{align}\label{sumA}
\notag\boxdot\big((A_\nu)_{ij}+(A_\nu)_{ji}\big)=&\;{(A^{[\mu})_{\nu]}}^ke_k\big((A_\mu)_{ij}+(A_\mu)_{ji}\big)
+e^\mu\big({(A_{[\mu})_{\nu]}}^k\big[(A_k)_{ij}+(A_k)_{ji}\big]\big)\\
&+e_\nu\big(A\big[(A)_{ij}+(A)_{ji}\big]\big)
+e_\nu\big({(A^\mu)_\mu}^k\big[(A_k)_{ij}+(A_k)_{ij}\big]\big),
\end{align}
where we have assumed that the sum ($A^2)_{ij}$+($A^2)_{ji}$ corresponding
to the term $A^2$ in the gauge condition (\ref{gauge}) can be expressed as
$A\big[(A)_{ij}+(A)_{ji}\big]$. Since (\ref{sumA}) has zero initial data (\ref{initAO}), the symmetric sums are zero everywhere and hence the skew-symmetry (\ref{Dg}) propagates.\footnote{This follows by a basic a priori energy estimate for linear systems like (\ref{sumA}), which in the singular Schwarzschild background is derived in \S\ref{Enest} for the more involved quasilinear system (\ref{system1}).} 
\begin{proof}[Proof of proposition \ref{IDprop}; EVE and Lorentz gauge]
Recall (\ref{redeq}) and the reduced equations $H_{\nu ij}=0$.
By assumption $(A_\nu)_{ij}$ is a solution of (\ref{boxA3}), i.e., the RHS of (\ref{redeq}) vanishes. Taking the divergence of (\ref{redeq}) with respect to the index $\nu$, the first part 
of the LHS of (\ref{redeq}), corresponding to the curl of the Ricci tensor, vanishes and we are left with
\begin{align}\label{boxgauge}
\square_g\big(\text{div}A_X-&\;A^2\big)_{ij}\\
\notag=&\;{(A^\nu)_i}^ce_\nu\big(\text{div}A_X-A^2\big)_{cj}+{(A^\nu)_j}^ce_\nu\big(\text{div}A_X-A^2\big)_{ic}.
\end{align}
The Lorentz gauge condition is valid initially (\ref{initA}). If the $e_0$ derivative of $\big(\text{div}A_X-A^2\big)_{ij}$ is zero as well on $\Sigma_0$, then the Lorentz gauge is valid in all of $\mathcal{M}^{1+3}=\Sigma\times[0,T]$.\footnote{Note however that the term $e_0\big(\text{div}A_X- A^2\big)_{ij}$ is of second order in $A$ and hence not at the level of initial data for (\ref{boxA3}), which we are allowed to prescribe. If zero initially, this should be a consequence of the geometric nature of the equations.} This is in fact implied by (\ref{redeq}), putting $\nu=0$ we have
\begin{align}\label{e_0gauge}
e_0\big(\text{div}A_X-A^2\big)_{ij}= \nabla_j\text{R}_{0i}- \nabla_i\text{R}_{0j}=0&&\text{on $\Sigma_0$}
\end{align}
by virtue of the vanishing of $\text{R}_{ab}(\tau=0)$ (\ref{initA}) and the (twice contracted) second Bianchi identity, $\nabla^a\text{R}_{ab}=\frac{1}{2}\text{R}$,  to replace if necessary a transversal derivative with tangential ones to $\Sigma_0$.

On the other hand, taking the $\nabla^i$ divergence of (\ref{redeq}) and commuting derivatives we obtain
\begin{align}\label{boxRic}
\notag\square_g\text{R}_{\nu j}=&\;\nabla^i\nabla_j\text{R}_{\nu i}=\frac{1}{2}\nabla_j\nabla_\nu \text{R}
+{{{\text{R}^i}_j}^c}_\nu \text{R}_{ci}+{{{\text{R}^i}_j}^c}_i\text{R}_{\nu c}\\
=&\;{{{\text{R}^i}_j}^c}_\nu \text{R}_{ci}+{\text{R}_j}^c\text{R}_{\nu c},
\end{align}
where we employed again the twice contracted second Bianchi identity and the fact that the scalar curvature $\text{R}$ vanishes everywhere: [contracting $\{\nu j\}$ in (\ref{redeq})]
\begin{align}\label{R=0}
0=\nabla_i\text{R}-\frac{1}{2}\nabla_i\text{R}=\frac{1}{2}\nabla_i\text{R},&&\text{R}\big|_{\Sigma_0}=0.
\end{align}
Now that we know the Lorentz gauge is valid, the identities (\ref{initA}) and (\ref{redeq}) $i=0$ imply
\begin{align}\label{e_0Ric}
\text{R}_{\nu j}=0,&&\nabla_0\text{R}_{\nu j}=\nabla_j \text{R}_{\nu 0},&&\text{on $\Sigma_0$}.
\end{align}
Utilizing the second Binachi identity $\nabla^a\text{R}_{ab}=\frac{1}{2}\text{R}=0$ once more, we conclude that $\nabla_0\text{R}_{\nu j}$ vanishes and hence the initial data set of (\ref{boxRic}) is the trivial one.
Thus, the initial condition $\text{R}_{\nu j}(\tau=0)=0$ (\ref{initA}) propagates and the spacetime $(\mathcal{M}^{1+3},g)$ obtained from the solution of (\ref{boxA3}) verifies the EVE (\ref{EVE}).
\end{proof}
\begin{remark}\label{initfixed}
Given the frame $\{e_i\}^3_0$ initially on $\Sigma_0$, and once the components $(A_0)_{ij}(\tau=0)$ have been chosen,\footnote{The $(A_0)_{ij}$'s are not fixed by the Lorentz gauge condition; cf. Lemma \ref{gaugemap}. They correspond to the $\partial_0$ derivative of the frame components $e_i$, which we can freely assign initially.} then the initial data set of (\ref{boxA}) is fixed by condition (\ref{initA}), i.e., the EVE and Lorentz gauge on $\Sigma_0$. Indeed, the components $(A_\nu)_{ij}(\tau=0)$, $\nu,i,j=1,2,3$, are determined uniquely by the orthonormal frame $\{e_i\}^3_1$ on $(\Sigma_0,\overline{g})$. The $(A_i)_{0j}(\tau=0)$'s correspond to the components of second fundamental form $K_{ij}$ of $\Sigma_0$, which is given by the solution to the constraints (\ref{const}), included in (\ref{initA}). Moreover, the expression of (\ref{initA}) in terms of $A$, for $\nu,i=1,2,3$, reads (schematically)
\begin{align}\label{initfixed1}
e_0(A_\nu)_{0i}=e_\nu(A)+A^2\\
\notag e_0(A_0)_{ij}=\sum_{\mu=1}^3e_\mu(A_\mu)_{ij}+A^2&&\text{on $\Sigma_0$}.
\end{align}
Hence, the LHS functions are expressed in terms of already determined components.
Finally, the rest components $e_0(A_\nu)_{ij}(\tau=0)$, $\nu,i,j=1,2,3$, are fixed by the algebraic property of the Riemann tensor 
\begin{align*}
\text{R}_{0\nu ij}=\text{R}_{ij0\nu}\\
e_0(A_\nu)_{ij}-e_\nu(A_0)_{ij}-([A_\mu,A_\nu])_{ij}-{(A_{[\mu})_{\nu]}}^k(A_k)_{ij}=\\
e_i(A_j)_{0\nu}-e_j(A_i)_{0\nu}-([A_i,A_j])_{0\nu}-{(A_{[i})_{j]}}^k(A_k)_{0\nu}
\end{align*}
or 
\begin{align}\label{initfixed2}
e_0(A_\nu)_{ij}=e_\nu(A_0)_{ij}+e_i(A_j)_{0j}-e_j(A_i)_{0\nu}+A^2,\qquad\text{on $\Sigma_0$},
\end{align}
since all the terms in the RHS have been accounted for. 
Notice that the definition of Riemann curvature was implicitly used in deriving (\ref{boxRic}) upon commuting covariant derivatives.
\end{remark}

\end{document}